\documentclass[12pt,a4paper]{article}
\usepackage{hyperref}
\usepackage{color}
\usepackage{lineno}
\usepackage{graphicx}
\usepackage{authblk}
\usepackage{amsmath,amsthm}
\usepackage{stmaryrd}
\usepackage{amsfonts}
\usepackage{amssymb}
\usepackage{mathrsfs}
\usepackage{subcaption}
\usepackage[top=2.5cm, bottom=2.5cm, left=2cm, right=2cm]{geometry}

\newcommand{\Fp}{\boldsymbol F_{\mathrm p}}

\newcommand{\Lp}{\boldsymbol L_{\mathrm p}}

\newcommand{\Qp}{\boldsymbol Q_{\mathrm p}}

\newtheorem{theorem}{\bf Theorem}[section]
\newtheorem{proposition}{\bf Proposition}[section]

\newtheorem{remark}{\bf Remark}[section]
\newtheorem{definition}{\bf Definition}[section]

\usepackage{lmodern}
\begin{document}
\allowdisplaybreaks

\title{\emph{A posteriori controllability} of non-linear remodelling in linear elastic continua}

\author{
Salvatore Di Stefano\footnote{Dipartimento di Matematica, University ``Aldo Moro'' of Bari - via Edoardo Orabona 4, 70125, Bari (BA), Italy} \hspace{0.15cm} and Marta Zoppello\footnote{Dipartimento di Scienze Matematiche ``Giuseppe Luigi Lagrange", Politecnico di Torino, corso Duca degli Abruzzi, 24, 10129, Torino (TO), Italy}}

\maketitle






\begin{abstract}
\noindent We investigate the \emph{a posteriori controllability} of remodelling within the framework of Geometric Control Theory applied to continuum systems. 

\noindent Remodelling is described through an isochoric Bilby--Kröner--Lee decomposition, where the internal structural transformation is represented by a remodelling tensor and its evolution is governed by a stress-driven constitutive law. Assuming infinitesimal deformation and remodelling stretches, while retaining finite rotations of the principal remodelling directions, we derive a nonlinear state-space representation in terms of principal remodelling stretches and orientation variables. This formulation allows the remodelling process to be interpreted as a finite-dimensional nonlinear control system parametrised by the material position. Equilibrium configurations are characterised, and under-actuated controllability problems are studied by means of the tools of Geometric Control Theory. The three-dimensional setting shows that the nonlinear coupling between remodelling stretches and orientations enables complete steering of the internal state through a reduced number of independent controls, while the planar reduction reveals a loss of controllability despite the presence of nonlinearity. \\
The results provide a first systematic investigation of remodelling controllability and further support the development of Continuum Control Theory as a framework for the analysis and design of controlled internal structural transformations in continuous media.
\end{abstract}

\section{Introduction}
Describing how external factors influence the evolution of a mechanical system is a challenging task, as it requires identifying the relevant physical variables and capturing their contribution to the evolution of the internal structure of the medium. This modelling process may be grounded on theoretical assumptions consistent with the underlying physics or inferred from experimental observations and generally involves mechanisms operating across different spatial and temporal scales. Examples include plastic rearrangement in non-living materials \cite{Lubliner2008a,Micunovic2009a,Cermelli2001a}, remodelling in biological tissues \cite{Ambrosi2011a,Ambrosi2019a}, cellular aggregates \cite{Giverso2012a,DiStefano2022c} and cell-matrix interactions \cite{DiStefano2026b}. In all these situations, deformation is accompanied by a structural rearrangement evolving at its own characteristic spatial and temporal scales, with the interaction between these two mechanisms being responsible for the overall mechanical response of the system.

\subsection{\emph{A priori} and \emph{a posteriori} strategies}
External processes may enter a continuum model through two conceptually distinct channels. First, they may be incorporated into an evolution law, thereby determining whether an internal, or structural, transformation is activated or inhibited and regulating the rate at which it proceeds \cite{Micunovic2009a,Maugin1998a}. Second, their action may be represented through generalized external forces that influence the evolution of a continuum system \cite{Cermelli2001a,dicarlo2002a}.

Focusing on the first channel, we distinguish between the \emph{a priori} and \emph{a posteriori} approaches \cite{GrilloMMS2023a,Grillo2023MEMOCSa,Grillo2023MEMOCSb}. In the former, the evolution law is assigned from the outset, possibly on the basis of phenomenological assumptions or experimental observations, whereas, in the latter, it is not prescribed in advance but emerges from the governing dynamics of the medium. This distinction was first developed in growth mechanics, where the evolution law entered the model as a time-dependent non-holonomic constraint on the (tensorial or scalar) parameter describing growth in the \emph{a priori} formulation, whereas no such constraint was prescribed in the \emph{a posteriori} formulation and the growth law emerged as part of the governing dynamics \cite{GrilloMMS2023a,Grillo2023MEMOCSa,Grillo2023MEMOCSb}. Although the terminology originated in the context of volumetric growth, it is entirely independent of mass-growth phenomena and can be extended to any internal degree of freedom describing a specific structural transformation.

Turning to the second channel, the key issue concerns the generalized force associated with external agents. When such a force is prescribed at the constitutive level, possibly depending on the mechanical state, the internal degrees of freedom, and non-mechanical factors acting on the medium, we speak of \emph{a priori controllability} \cite{DiStefano2025a}. In such a case, by specifying the form of the force, controlling the system means controlling the material parameters appearing in it \cite{Ambrosi2025a}. Conversely, when it is instead reconstructed from the state that the system is required to reach, we have the notion of \emph{a posteriori controllability} \cite{DiStefano2025a}. The latter perspective, first introduced in \cite{DiStefano2025a} in the context of volumetric growth, is naturally rooted in Geometric Control Theory, where external actions are interpreted as control functions \cite{Coron2007} and their admissible forms are determined according to controllability arguments based on their ability to steer the \emph{structural} state of a system toward a prescribed one. We also refer to \cite{DiStefano2026a} for an application of the \emph{a posteriori controllability} to adhesive composites. 

\subsection{From \emph{Geometric} to \emph{Continuum Control Theory}}
Geometric Control Theory provides a mathematical framework for investigating whether the evolution of a dynamical system can be steered toward prescribed states through the action of suitably chosen controls \cite{AgrachevBook,Coron2007,LibroJurdjevic}. Rather than focusing exclusively on the solution of a given evolution problem, it addresses the possibility of driving a system between different configurations and characterizes the geometric and analytical structure underlying such processes.  Traditionally, Geometric Control Theory has been successfully applied to finite-dimensional systems arising in mechanics \cite{bloch,CardinGianniottiSpiro2021}, robotics \cite{SansonettoZoppello2020,Wang2020OptimalCO}, locomotion \cite{MMSZ,FPZ} and engineering applications \cite{BlochObstacleAvoidance,DMPSansonetto}. In these settings, the state variables describe kinematic or dynamical quantities, whereas the controls are interpreted as external actions acting on the system. The central question concerns the possibility of reaching desired states, configurations or trajectories through suitable control inputs.

In the present work, we adopt the term \emph{Continuum Control Theory} to denote the modelling framework in which the methods of Geometric Control Theory are applied to investigate and steer the spatial and/or temporal evolution of continuum systems. To this end, and to the best of our knowledge, the present work, together with \cite{DiStefano2025a}, constitutes one of the first attempt to systematically apply Geometric Control Theory to continuum mechanics. For this reason, we focus on continuum systems whose evolution can be formulated as a system of ordinary differential equations parametrized by space, as done in \cite{DiStefano2025a} for volumetric growth. The extension of the proposed framework to a fully partial differential equation setting is left to future work. The significance of that approach, however, does not lie in its application to growth alone. Rather, it relies on the more general observation that several continuum theories involve internal variables whose evolution is governed by ordinary differential equations and can therefore be interpreted as finite-dimensional control systems after suitable parametrizations.

\subsection{Geometric Control Theory of remodelling} We understand remodelling as the time-dependent reorganization of the internal structure of a material, accompanied by an evolution of its local mechanical properties. In non-living materials, such structural changes are often associated with plastic rearrangements and irreversible modifications of the material architecture \cite{Lubliner2008a,Micunovic2009a}, whereas in living systems they may arise from biological processes such as adaptation, turnover, and cellular reorganization. \cite{Lubliner2008a,Ambrosi2011a,Ambrosi2019a,Micunovic2009a}. In the framework adopted here, it is represented by an isochoric second-order tensor that transforms the local material organization while preserving volume. Through the Bilby--Kröner--Lee (BKL) decomposition \cite{Sadik2017a,Rodriguez1994a,Preston2010}, this internal transformation is separated from the elastic accommodation required to embed the remodelled natural state into a compatible current configuration. In general, the interactions driving the evolution of the internal structure may not be known with sufficient detail to prescribe their functional form, and their experimental identification may require relating observations across different spatial and temporal scales \cite{Ambrosi2011a,Ambrosi2019a,Epstein2015a,Garikipati2006a,Latorre2018a,Latorre2020a}. In conventional remodelling theories, these interactions are encoded through evolution laws \cite{Epstein2007a,Micunovic2009a}, while a control-theoretic viewpoint introduces an additional level of description by asking whether the resulting remodelling process can be actively steered toward prescribed states. Following \cite{DiStefano2025a}, we here extend the same conceptual framework to remodelling. The objective is to investigate whether the internal variables governing remodelling can be steered toward prescribed states through suitable generalized external actions. In this setting, the remodelling law defines the uncontrolled dynamics, or drift, of the system, whereas additional control functions represent external actions that are not already contained within that law \cite{DiStefano2025a}. These controls act directly on the remodelling dynamics and are selected in order to drive the internal state toward a prescribed target configuration or along a desired trajectory. Accordingly, \emph{a posteriori controllability} of remodelling consists in reconstructing the generalized actions required to realize a prescribed reorganization of the material structure.

\paragraph{Possible experimental implementation of remodelling control.} A control action should be understood as an external intervention capable of modifying the evolution of the internal material structure beyond the constitutive response generated by the local stress state alone. In biological tissues, such actions may be associated with biochemical stimuli, growth factors, pharmacological treatments, electrical stimulation, or controlled mechanical loading protocols capable of influencing cell activity, matrix turnover, and fibre reorientation \cite{Ambrosi2019a,Humphrey2014a,Chen2019a}. In non-living materials, possible control mechanisms include thermal treatments, electromagnetic fields, controlled plastic deformation, residual-stress engineering, microstructural processing, or external stimuli capable of inducing preferential rearrangements of the internal architecture \cite{Qi2022a,Liu2024a}. More generally, advances in programmable matter, smart materials, and architectured metamaterials increasingly allow external fields and boundary actions to influence internal structural transformations in a controlled manner. 

Within the perspective of \emph{a posteriori controllability}, the objective is not to prescribe a specific physical mechanism a priori, but rather to determine the generalized actions required to achieve a desired remodelling evolution. Once such actions have been reconstructed through the control-theoretic framework, the associated controls may be interpreted as target inputs for experimental design. Establishing a correspondence between the mathematical controls and physically realizable stimuli is therefore a natural direction for future research.

\paragraph{Manuscript organization.} The manuscript is organized as follows. Section \ref{sec_thFram} introduces the kinematic, constitutive and balance equations governing deformation and remodelling. Section \ref{sec_lin} derives the linearized remodelling equations while retaining finite rotations of the principal remodelling directions. Section \ref{sec_state} develops the corresponding nonlinear state-space formulation and studies its equilibrium configurations and regularity properties. Section \ref{sec_control} reviews the main tools of Geometric Control Theory employed throughout the paper. Section \ref{sec_under} investigates under-actuated controllability for selected three-dimensional remodelling regimes, whereas Section \ref{sec:planar_state_space_equation} addresses the planar reduction and establishes small-time local controllability through different actuation strategies. Final remarks and future perspectives are collected in Section \ref{sec_concl}.

\section{Theoretical framework for deformation and remodelling}
\label{sec_thFram}
We consider a single-phase solid body subject to deformation and remodelling, two distinct but mutually coupled physical processes. Deformation denotes changes in the placement and shape of the body at the continuum scale, while remodelling refers to the time‑dependent reorganisation of its internal structure and the consequent evolution of local material properties. On this basis, we introduce the notation and the mathematical framework used throughout the paper, formulate the relevant kinematics, and state the constitutive assumptions that govern the medium’s mechanical response.

\subsection{Configurations, natural state and kinematics}
Let $\mathscr{B}\subset\mathscr{S}$ be the reference configuration of the considered continuum body, where $\mathscr{S}$ is the three-dimensional Euclidean space. We equip $\mathscr{B}$ with the metric induced by $\mathscr{S}$ and adopt oriented Cartesian coordinates, so that each material point $X\in\mathscr{B}$ and $x\in\mathscr{S}$ are labelled as
\begin{linenomath}
\begin{align}
&X=(X_1,X_2,X_3)\in\mathbb{R}^3 \quad \mbox{and} \quad x=(x_1,x_2,x_3)\in\mathbb{R}^3,
\end{align}
\end{linenomath}
where, for notational simplicity, the same symbols $X$ or $x$ denote both material points and their coordinate representation. The motion of $\mathscr{B}$ is a one-parameter family of embeddings
\begin{linenomath}
\begin{align}
\chi(\cdot,t):\mathscr{B}\longrightarrow\mathscr{S},
\end{align}
\end{linenomath}
parametrised by time $t\in\mathscr{I}$, where $\mathscr{I}=[t_{\rm in}, t_{\rm fin}]\subset\mathbb{R}$ is the time window and the image 
\begin{linenomath} 
\begin{align} \mathscr{B}_t:=\chi(\mathscr{B},t) \end{align} 
\end{linenomath} 
is the current configuration of $\mathscr{B}$ at time $t\in\mathscr{I}$. In particular, the material point $X\in\mathscr{B}$ is mapped into $x=\chi(X,t)\in\mathscr{B}_t$ by $\chi$ at time $t\in\mathscr{I}$. Moreover, we denote by $T_X\mathscr{B}$ and $T_x\mathscr{S}$ the tangent spaces to $\mathscr{B}$ at $X$ and to $\mathscr{S}$ at $x$, respectively, and by $T_X^*\mathscr{B}$ and $T_x^*\mathscr{S}$ the corresponding cotangent spaces. Finally, we introduce two positively oriented orthonormal (Cartesian) bases
\begin{linenomath}
\begin{align}
\{\boldsymbol{G}_A(X)\}_{A=1,2,3}
\quad\text{and}\quad
\{\boldsymbol{g}_a(x)\}_{a=1,2,3},
\label{eq_bases}
\end{align}
\end{linenomath}
of $T_X\mathscr{B}$ and $T_x\mathscr{S}$, respectively, with $X\in\mathscr{B}$ and $x\in\mathscr{S}$. From here on, uppercase Latin indices refer to the material quantities, i.e., associated with $\mathscr{B}$, while lowercase Latin indices indicate spatial quantities, namely defined on $\mathscr{S}$. 

The deformation gradient tensor is the tangent map of $\chi(\cdot,t)$ at $X\in\mathscr{B}$ and $t\in\mathscr{I}$, i.e., 
\begin{linenomath}
\begin{align}
\boldsymbol{F}(X,t):=T\chi(X,t):T_{X}\mathscr{B}\rightarrow T_{\chi(X,t)}\mathscr{S}, 
\end{align}
\end{linenomath}
whose components are given by
\begin{linenomath}
\begin{align}
[\boldsymbol{F}(X,t)]_{aA}=\dfrac{\partial \chi_{a}}{\partial X_{A}}(X,t),
\end{align}
\end{linenomath}
and since the motion is assumed to be orientation preserving, the (material) volumetric ratio
\begin{linenomath}
\begin{align}
J(X,t):=\det\boldsymbol{F}(X,t)>0.
\label{eq_vol_ratio}
\end{align}
\end{linenomath}
is strictly positive, for any $(X,t)\in\mathscr{B}\times\mathscr{I}$. We also introduce the right Cauchy--Green deformation tensor 
\begin{linenomath} 
\begin{align} 
\boldsymbol{C}(X,t) := \boldsymbol{F}^{\mathrm{T}}(X,t) \boldsymbol{F}(X,t) : T_{X}\mathscr{B} \longrightarrow T_{X}\mathscr{B},
\end{align} 
\end{linenomath} 
where the superscript $^\mathrm T$ denotes the metric adjoint with respect to the Euclidean inner products on the relevant domain and codomain.

We describe remodelling through a second-order tensor field, referred to as the \emph{remodelling tensor} \cite{Ambrosi2011a,Ambrosi2019a,Epstein2007a,Micunovic2009a}, which is defined over $\mathscr{B}\times\mathscr{I}$ and denoted by $\boldsymbol{F}_{\hspace{-0.5mm}\mathrm{p}}$. The introduction of $\boldsymbol{F}_{\hspace{-0.5mm}\mathrm{p}}$ stems from the Bilby--Kr\"oner--Lee (BKL) decomposition of the deformation gradient tensor \cite{Sadik2017a,Rodriguez1994a,Ciancio2008a,Preston2010}, according to which 
\begin{linenomath} 
\begin{align} 
\boldsymbol{F}(X,t)= \boldsymbol{F}_{\hspace{-0.5mm}\mathrm{e}}(X,t)\boldsymbol{F}_{\hspace{-0.5mm}\mathrm{p}}(X,t), \qquad (X,t)\in\mathscr{B}\times\mathscr{I}, \label{def_BKL_decomp} 
\end{align} 
\end{linenomath} 
where $\boldsymbol{F}_{\hspace{-0.5mm}\mathrm{e}}$ is the \emph{elastic} or \emph{accommodating tensor}.

\paragraph{A brief overview of the BKL decomposition.} According to the multiplicative splitting in Equation~\eqref{def_BKL_decomp}, the overall deformation of $\mathscr{B}$ results from two physically distinct contributions. The tensor $\boldsymbol{F}_{\mathrm p}$ describes changes in the material organisation induced by remodelling, whereas $\boldsymbol{F}_{\mathrm e}$ accounts for the elastic accommodation required to realise the remodelled state within the current configuration \cite{Sadik2017a,Rodriguez1994a,Ciancio2008a,Preston2010,Micunovic2009a}. In general, neither $\boldsymbol{F}_{\mathrm p}$ nor $\boldsymbol{F}_{\mathrm e}$ is integrable, that is, neither can be identified with the gradient of a motion-like map \cite{Ciancio2008a,Preston2010,Micunovic2009a} and Equation~\eqref{def_BKL_decomp} is understood locally, as a factorisation acting on each material, or body, element, rather than as a global transformation of $\mathscr{B}$. Following \cite{Ciancio2008a,Preston2010,Micunovic2009a}, $\boldsymbol{F}_{\mathrm p}$ may be viewed as mapping each body element into a locally relaxed, stress-free state, referred to as its \emph{natural state}. This state is conceptually obtained by isolating an infinitesimal material portion from its surroundings and allowing it to relax independently. However, the natural states of neighbouring body elements cannot be assembled into a globally compatible, stress-free configuration. The elastic distortions represented by $\boldsymbol{F}_{\mathrm e}$ are therefore required to accommodate these local states within a compatible current configuration. When the local natural states are incompatible, such elastic accommodation may give rise to residual stresses, even in the absence of external loads \cite{Micunovic2009a}.

Using a standard notation, for each $X\in\mathscr{B}$ and $t\in\mathscr{I}$, we postulate the existence of a vector space $\mathscr{N}_X(t)$ representing the locally relaxed, stress-free counterpart of $T_X\mathscr{B}$ \cite{Ciancio2008a,Preston2010,Micunovic2009a}. The remodelling tensor maps the tangent space of the reference configuration to the local natural space, i.e., 
\begin{linenomath} 
\begin{align} 
\boldsymbol{F}_{\hspace{-0.5mm}\mathrm{p}}(X,t): T_X\mathscr{B} \longrightarrow \mathscr{N}_X(t). \end{align} 
\end{linenomath} 
In turn, the accommodating tensor maps the local natural space into the tangent space to the current configuration, namely, 
\begin{linenomath} 
\begin{align} 
\boldsymbol{F}_{\hspace{-0.5mm}\mathrm{e}}(X,t): \mathscr{N}_X(t) \longrightarrow T_{\chi(X,t)}\mathscr{S}. \end{align} 
\end{linenomath} 
Accordingly, the multiplicative decomposition may be represented locally as 
\begin{linenomath} 
\begin{align} 
T_X\mathscr{B} \xrightarrow{\;\boldsymbol{F}_{\hspace{-0.5mm}\mathrm{p}}(X,t)\;} \mathscr{N}_X(t) \xrightarrow{\;\boldsymbol{F}_{\hspace{-0.5mm}\mathrm{e}}(X,t)\;} T_{\chi(X,t)}\mathscr{S}, 
\end{align} 
\end{linenomath} 
for every $X\in\mathscr{B}$ and $t\in\mathscr{I}$. The family of local natural spaces $\mathscr{N}_X(t)$ over $X\in\mathscr{B}$ defines the natural state of the continuum at time $t$, i.e.,
\begin{linenomath} 
\begin{align} 
\mathscr{N}(t)=\bigsqcup_{X\in\mathscr{B}} \{X\}\times\mathscr{N}_X(t), \qquad t\in \mathscr{I}.
\end{align} 
\end{linenomath} 
Analogously to Equation \eqref{eq_bases}, for each $(X,t)\in\mathscr{B}\times\mathscr{I}$, we choose a positively oriented orthonormal basis
\begin{linenomath}
\begin{align}
\{\mathfrak{g}_\alpha(X,t)\}_{\alpha=1,2,3}
\label{eq_bases2}
\end{align}
\end{linenomath}
of $\mathscr{N}_X(t)$ and with Greek indices denoting physical quantities related to the natural state. We remark that, the selected orthonormal bases in Equations \eqref{eq_bases} and \eqref{eq_bases2} identify the three Euclidean spaces $T_{X}\mathscr{B}$, $T_{x}\mathscr{S}$ and $\mathscr{N}_X(t)$ with $\mathbb{R}^3$, for each $(X,t)\in\mathscr{B}\times\mathscr{I}$, while preserving their distinct geometric roles. Hereafter, whenever no confusion can arise, the same symbols are used for tensors and for their matrix representations in these bases. 

\medskip
\noindent
By virtue of Equation~\eqref{def_BKL_decomp}, the material volumetric ratio in Equation \eqref{eq_vol_ratio} admits the multiplicative decomposition $J=J_{\mathrm{e}}J_{\mathrm{p}}$, where $J_{\mathrm e}:=\det\boldsymbol{F}_{\mathrm e}$ and $J_{\mathrm p}:=\det\boldsymbol{F}_{\mathrm p}$ quantify the elastic and remodelling-induced volume changes, respectively. We assume remodelling to induce isochoric distortions in $\mathscr{B}$, so that 
\begin{linenomath} 
\begin{align} 
J_{\mathrm p}(X,t)=1, \qquad \text{for all }(X,t)\in\mathscr{B}\times\mathscr{I}.
\label{iso_constr}
\end{align} 
\end{linenomath} 
This hypothesis reflects the fact that the present work is concerned with remodelling processes that reorganize the internal material architecture without involving volumetric growth or mass production. For later use, we introduce the elastic right Cauchy--Green tensor 
\begin{linenomath} 
\begin{align} 
\boldsymbol{C}_{\hspace{-0.5mm}\mathrm{e}}(X,t) := \boldsymbol{F}_{\hspace{-0.5mm}\mathrm{e}}^{\mathrm{T}}(X,t) \boldsymbol{F}_{\hspace{-0.5mm}\mathrm{e}}(X,t) : \mathscr{N}_X(t) \longrightarrow \mathscr{N}_X(t), \end{align} 
\end{linenomath} 
and the remodelling right Cauchy--Green tensor
\begin{linenomath} 
\begin{align} 
\boldsymbol{C}_{\hspace{-0.5mm}\mathrm{p}}(X,t) := \boldsymbol{F}_{\hspace{-0.5mm}\mathrm{p}}^{\mathrm{T}}(X,t) \boldsymbol{F}_{\hspace{-0.5mm}\mathrm{p}}(X,t) : T_X\mathscr{B} \longrightarrow T_X\mathscr{B}, \end{align} 
\end{linenomath} 
for every $(X,t)\in\mathscr{B}\times\mathscr{I}$.

\begin{remark}[Euclidean representation]
\label{rem_metric}
The tangent spaces $T_X\mathscr{B}$ and $T_x\mathscr{S}$ and the natural spaces $\mathscr{N}_X(t)$ are endowed with Euclidean inner products. The bases in Equations \eqref{eq_bases} and \eqref{eq_bases2} are positively oriented and orthonormal with respect to these inner products. Consequently, the metric adjoint is represented by the ordinary matrix transpose. The identity tensors acting on $T_X\mathscr{B}$, $T_x\mathscr{S}$, and $\mathscr{N}_X(t)$ remain geometrically distinct, although each is represented by the identity matrix in the corresponding orthonormal basis. Finally, time derivatives of tensors are understood as derivatives of their matrix representatives in this selected Euclidean trivialisation. 
\end{remark}

\subsection{Constitutive framework}
We assume a hyperelastic constitutive response with the remodelling tensor $\boldsymbol F_{\rm p}$ acting as an internal degree of freedom \cite{dicarlo2002a}. Accordingly, there exists a strain-energy density per unit volume of the reference configuration, denoted by $\psi_{\mathrm R}$, which characterises the mechanical response of the body. We initially admit the constitutive representation
\begin{linenomath}
\begin{align}
\psi_{\mathrm{R}}(X,t)=\check{\psi}_{\mathrm{R}}(\boldsymbol{F}(X,t),X,t)
\end{align}
\end{linenomath}
We remark that the explicit dependence on $X\in\mathscr{B}$ and $t\in\mathscr{I}$ allows the material properties to vary in space and time as a result of remodelling (see, for instance, \cite{Epstein2007a}). To distinguish between changes in material response induced by remodelling and the intrinsic constitutive behaviour of the material itself, we introduce a strain-energy density $\check\psi_{\rm n}$, defined per unit volume of the local natural state, such that, consistently with the material-uniformity assumption (see \cite{Epstein2007a} and references therein)
\begin{linenomath}
\begin{align}
\psi_{\mathrm{R}}(X,t)=\check{\psi}_{\mathrm{R}}(\boldsymbol{F}(X,t),X,t)=J_{\mathrm{p}}\check{\psi}_{\mathrm{n}}(\boldsymbol{F}(X,t)\boldsymbol{F}_{\hspace{-0.5mm}\mathrm{p}}^{-1}(X,t))\equiv J_{\mathrm{p}}\check{\psi}_{\mathrm{n}}(\boldsymbol{F}_{\hspace{-0.5mm}\mathrm{e}}(X,t)).
\end{align}
\end{linenomath}
Thus, any explicit dependence of the referential strain-energy density on the material point and time is mediated by the remodelling tensor. In particular, since $J_{\mathrm{p}}=1$, we obtain
\begin{linenomath}
\begin{align}
\psi_{\mathrm{R}}(X,t)=\check{\psi}_{\mathrm{R}}(\boldsymbol{F}(X,t),X,t)=\check{\psi}_{\mathrm{n}}(\boldsymbol{F}(X,t)\boldsymbol{F}_{\hspace{-0.5mm}\mathrm{p}}^{-1}(X,t))\equiv \check{\psi}_{\mathrm{n}}(\boldsymbol{F}_{\hspace{-0.5mm}\mathrm{e}}(X,t)).
\end{align}
\end{linenomath}
Thus, the dependence of the strain-energy density on remodelling is entirely mediated by the change in local material architecture encoded by $\boldsymbol F_{\mathrm p}$ and not by volumetric effects. For the present work, we adopt an isotropic Saint-Venant--Kirchhoff strain-energy density in the natural state of the form
\begin{linenomath}
\begin{align}
\check{\psi}_{\mathrm{n}}(\boldsymbol{E}_{\mathrm{e}})=\frac{1}{2}\lambda[\mathrm{tr}(\boldsymbol{E}_{\mathrm{e}})]^2+\mu\boldsymbol{E}_{\mathrm{e}}\cdot\boldsymbol{E}_{\mathrm{e}},
\label{psi_n}
\end{align}
\end{linenomath} 
where the dependence on $\boldsymbol{F}_{\mathrm{e}}$ is given by the elastic Green-Lagrange distortion tensor
\begin{linenomath}
\begin{align}
\boldsymbol{E}_{\mathrm{e}}=\tfrac{1}{2}\left(\boldsymbol{C}_{\mathrm{e}}-\boldsymbol{I}\right),
\label{eq_Ee}
\end{align}
\end{linenomath} 
with $\lambda$ and $\mu$ being the Lam\'e parameters of isotropic linear elasticity and $\boldsymbol{I}$ the identity tensor\footnote{Although the identity tensors acting on $T_X\mathscr{B}$, $T_x\mathscr{S}$, and $\mathscr{N}_X(t)$ are geometrically distinct, for every $X\in\mathscr{B}$ and $t\in\mathscr{I}$, they are represented by the identity matrix in the corresponding orthonormal basis of Equations \eqref{eq_bases} and \eqref{eq_bases2} (see Remark \ref{rem_metric}). Thus, we prefer to not introduce different notations and use the same symbol $\boldsymbol{I}$ independently on the space of definition of the considered identity tensor.}. The elastic second Piola--Kirchhoff stress tensor $\boldsymbol{S}_{\mathrm{e}}$, referred to the natural state, is then obtained by differentiation of the strain-energy density $\check{\psi}_{\rm n}$ in Equation \eqref{psi_n} with respect to $\boldsymbol{E}_{\mathrm{e}}$, so that
\begin{linenomath}
\begin{align}
\boldsymbol{S}_{\mathrm{e}}=\lambda\mathrm{tr}(\boldsymbol{E}_{\mathrm{e}})\boldsymbol{I}+2\mu\boldsymbol{E}_{\mathrm{e}}.
\label{eq_elastic_S}
\end{align}
\end{linenomath} 
Pulling $\boldsymbol{S}_{\mathrm e}$ back through the remodelling tensor, we define the second Piola--Kirchhoff stress in the reference configuration by
\begin{linenomath}
\begin{align}
\boldsymbol{S}
:=J_{\rm p}\boldsymbol{F}_{\mathrm p}^{-1}
\boldsymbol{S}_{\mathrm e}
\boldsymbol{F}_{\mathrm p}^{-\mathrm T}.
\end{align}
\end{linenomath}
and the corresponding first Piola--Kirchhoff stress as
\begin{linenomath}
\begin{align}
\boldsymbol{P}
=\boldsymbol{F}\boldsymbol{S}
=\boldsymbol{F}_{\mathrm e}
\boldsymbol{S}_{\mathrm e}
\boldsymbol{F}_{\mathrm p}^{-\mathrm T}.
\label{first_Piola}
\end{align}
\end{linenomath}
Finally, we introduce the material Mandel-type stress
\begin{linenomath}
\begin{align}
\boldsymbol{\Sigma}:=\boldsymbol{C}\boldsymbol{S}.
\label{eq_Mandel}
\end{align}
\end{linenomath}
Before proceeding, we remark that the Saint-Venant--Kirchhoff model is adopted solely because the subsequent analysis is restricted to the infinitesimal-strain regime. No claim is therefore made regarding its suitability for finite elastic strains.

\subsection{Balance of linear momentum and remodelling law}
Under the assumption of negligible inertia, the local balance of linear momentum referred to the reference configuration takes the form
\begin{linenomath}
\begin{align}
\mathrm{Div}\boldsymbol{P}=\boldsymbol{f}, && \textrm{in\,} \mathscr{B}\times \mathscr{I},
\label{eq_balance_linear_mom}
\end{align}
\end{linenomath}
where  $\boldsymbol{P}$ is the first Piola-Kirchhoff stress tensor introduced in Equation \eqref{first_Piola} and $\boldsymbol{f}$ are body forces. If $\boldsymbol{N}$ denotes the outward unit normal field of the boundary $\partial\mathscr{B}$ of $\mathscr{B}$, we supplement Equation \eqref{eq_balance_linear_mom} with the boundary conditions
\begin{linenomath}
\begin{subequations}
\begin{align}
&\boldsymbol{P}\boldsymbol{N}=\boldsymbol{\tau}, && \textrm{on\,} \mathscr{\partial_{\mathrm{N}}B}\times \mathscr{I},
\label{BVP_tractions}\\
&\chi=\chi_{\mathrm{b}}, && \textrm{on\,} \mathscr{\partial_{\mathrm{D}}B}\times \mathscr{I}, 
\label{BVP_Dirichlet}
\end{align}
\end{subequations}
\end{linenomath}
where the surface traction $\boldsymbol{\tau}$ is prescribed on the Neumann part $\partial_{\mathrm{N}}\mathscr{B}$ of $\partial\mathscr{B}$, whereas the motion $\chi_{\mathrm{b}}$ is assigned on its Dirichlet part $\partial_{\mathrm{D}}\mathscr{B}$. We further note that $\partial_{\mathrm{N}}\mathscr{B}$ and $\partial_{\mathrm{D}}\mathscr{B}$ form a topological partition of $\partial\mathscr{B}$, namely $\partial_{\mathrm{N}}\mathscr{B}\cup\partial_{\mathrm{D}}\mathscr{B}=\partial\mathscr{B}$ and $\overline{\partial_{\mathrm{N}}\mathscr{B}}\cap\partial_{\mathrm{D}}\mathscr{B}=\partial_{\mathrm{N}}\mathscr{B}\cap\overline{\partial_{\mathrm{D}}\mathscr{B}}=\emptyset$, where an overbar denotes the topological closure of a set. 

The boundary-value problem (BVP) defined by Equations \eqref{eq_balance_linear_mom}, \eqref{BVP_tractions} and \eqref{BVP_Dirichlet} is then complemented by an appropriate evolution equation for $\Fp$, or remodelling law, and a suitable initial condition for it. Following the literature, it is convenient to introduce the remodelling velocity tensor $\Lp$ as
\begin{linenomath}
\begin{align}
\boldsymbol{L}_{\hspace{-0.5mm}\mathrm{p}}
:=
\boldsymbol{F}_{\hspace{-0.5mm}\mathrm{p}}^{-1}
\dot{\boldsymbol{F}}_{\hspace{-0.5mm}\mathrm{p}},
\label{def_Lp}
\end{align}
\end{linenomath}
where a superimposed dot means time derivative and, because of the isochoricity constraint in Equation \eqref{iso_constr}, $\boldsymbol{L}_{\hspace{-0.5mm}\mathrm{p}}$ is a traceless tensor, since
\begin{linenomath}
\begin{align}
0=\dot J_{\mathrm p}
=J_{\mathrm p}\operatorname{tr}
\left(
\boldsymbol{F}_{\mathrm p}^{-1}
\dot{\boldsymbol{F}}_{\mathrm p}
\right)
=J_{\mathrm p}\operatorname{tr}(\boldsymbol{L}_{\mathrm p}).
\end{align}
\end{linenomath}
and, hence, $\boldsymbol L_{\mathrm p}$ possesses only eight independent components. In particular, by taking inspiration from \cite{Grillo2023MEMOCSa,Grillo2023MEMOCSb}, we write
\begin{linenomath}
\begin{align}
b(\boldsymbol{C}\Lp\boldsymbol{C}^{-1}+\Lp^{\rm T})+c(\boldsymbol{C}\Lp\boldsymbol{C}^{-1}-\Lp^{\rm T})=\mathrm{dev}\boldsymbol{\Sigma},
\label{rem_law_1}
\end{align}
\end{linenomath}
where $b$ and $c$ are strictly positive constitutive parameters having the physical units of a viscosity. In Equation \eqref{rem_law_1}, we assume the remodelling process to be driven by the deviatoric part of the Mandel stress tensor $\boldsymbol{\Sigma}$ in Equation \eqref{eq_Mandel}. This implies that it associates the evolution of the natural state with directional stress differences, namely, with distortional, rather than volumetric, components of the stress state. It is worth noting that the deviatoric projection removes the trace of $\boldsymbol{\Sigma}$ but does not, in general, enforce its symmetry. Hence, if $\boldsymbol{\Sigma}$ is nonsymmetric, $\mathrm{dev}\boldsymbol{\Sigma}$ may retain both symmetric and skew-symmetric contributions. Moreover, we have that the left-hand side is traceless, consistently with the fact that the right-hand side in Equation \eqref{rem_law_1} is purely deviatoric. Finally, we set the initial condition for $\Fp$ as
\begin{linenomath}
\begin{align}
\Fp(X,0)=\Fp^{0}(X), \qquad X\in\mathscr{B}.
\label{initial_condition_Fp}
\end{align}
\end{linenomath}
\paragraph{Mandel stress and control actions.} In the present framework, $\mathrm{dev}\boldsymbol{\Sigma}$ acts as a constitutive driving force governing the evolution of the natural state \cite{Cleja2000a}. Through Equation~\eqref{rem_law_1}, it determines the uncontrolled remodelling dynamics and therefore defines the drift term of the corresponding state-space system \cite{Coron2007,DiStefano2025a}. The control functions, instead, have a fundamentally different role, since they represent additional generalized actions superposed to the evolution law and reconstructed through the controllability analysis.

\subsection{Parametrisation of remodelling}
Since $\boldsymbol{F}_{\hspace{-0.5mm}\mathrm{p}}$ is unimodular, and hence non-singular, it admits the right polar decomposition 
\begin{linenomath} 
\begin{align} 
\boldsymbol{F}_{\mathrm p} = \boldsymbol{R}_{\mathrm p}\boldsymbol{U}_{\mathrm p}, 
\end{align} 
\end{linenomath} 
where $\boldsymbol{R}_{\mathrm p}$ is a proper rotation tensor and $\boldsymbol{U}_{\mathrm p}$ is symmetric, positive definite, and unimodular. By the spectral theorem, there exist a rotation tensor $\boldsymbol{Q}_{\mathrm p}$ and a positive-definite, diagonal, unimodular tensor $\boldsymbol{D}_{\mathrm p}$ such that
\begin{linenomath}
\begin{align}
\boldsymbol{U}_{\mathrm p} = \boldsymbol{Q}_{\mathrm p} \boldsymbol{D}_{\mathrm p} \boldsymbol{Q}_{\mathrm p}^{\mathrm T} \quad \mbox{and, therefore, } \quad \boldsymbol{F}_{\mathrm p}
=\boldsymbol{R}_{\mathrm p}
\left(
\boldsymbol{Q}_{\mathrm p}
\boldsymbol{D}_{\mathrm p}
\boldsymbol{Q}_{\mathrm p}^{\mathrm T}
\right).
\label{eq_polar_dec}
\end{align}
\end{linenomath}
This decomposition separates changes in the intensity of remodelling, encoded by the principal stretches, from changes in the orientation of the principal remodelling directions. To automatically satisfy positivity of the principal stretches and simplify the subsequent linearisation, we introduce the logarithmic variables $\ell_1,\ell_2:\mathscr{B}\times\mathscr{I}\to\mathbb{R}$ and set
\begin{linenomath}
\begin{align}
\boldsymbol{D}_{\mathrm p}
:=\operatorname{diag}
\left\{
\mathrm e^{\ell_1},
\mathrm e^{\ell_2},
\mathrm e^{-\ell_1-\ell_2}
\right\},
\label{eq_principal_str}
\end{align}
\end{linenomath}
with the third logarithmic stretch being $\ell_{3}=-\ell_1-\ell_2$ as a consequence of isochoricity. The rotation tensor $\Qp$ determines the orientation of
the principal remodelling directions in the reference configuration with respect to the adopted material
basis. To provide an explicit representation of this tensor, we introduce
the angle fields $\alpha,\beta,\gamma:
\mathscr{B}\times\mathscr{I} \longrightarrow \mathbb{R}$, and parametrise $\Qp$ through a proper Euler-angle representation, namely
\begin{linenomath} 
\begin{align} 
\Qp\hspace{-0.5mm}=\hspace{-0.5mm}\begin{pmatrix} \cos\alpha\cos\beta\cos\gamma-\sin\alpha\sin\gamma & -\cos\alpha\cos\beta\sin\gamma-\sin\alpha\cos\gamma & \cos\alpha\sin\beta \\ \sin\alpha\cos\beta\cos\gamma+\cos\alpha\sin\gamma & -\sin\alpha\cos\beta\sin\gamma+\cos\alpha\cos\gamma & \sin\alpha\sin\beta \\ -\sin\beta\cos\gamma & \sin\beta\sin\gamma & \cos\beta \end{pmatrix}.
\label{eq_Qp_explicit} 
\end{align} 
\end{linenomath}
We finally remark that the angle representation is local and is not globally unique. In
fact, when
\begin{linenomath}
\begin{align}
\sin\beta=0,
\qquad\text{that is,}\qquad
\beta=k\pi,
\quad k\in\mathbb{Z},
\end{align}
\end{linenomath}
the first and third rotations occur about the same effective axis, and the angles $\alpha$ and $\gamma$ cannot be determined independently. This is the standard gimbal-lock singularity associated with the Euler-angle representation: the rotation tensor $\Qp$ remains regular while its representation in terms of the selected angles becomes singular.

As a consequence of Equations \eqref{eq_polar_dec}, \eqref{eq_principal_str} and \eqref{eq_Qp_explicit}, the remodelling tensor is completely described by the finite set of variables $(\ell_1,\ell_2,\alpha,\beta,\gamma)$ together with the rotational contribution associated with $\boldsymbol R_{\mathrm p}$. Accordingly, the tensorial remodelling law in Equation \eqref{rem_law_1} can be reformulated as a nonlinear system of ordinary differential equations for these variables. This representation provides the basis for the state-space formulation developed in the following sections and for the subsequent controllability analysis. Also the initial condition for $\Fp$ in Equation \eqref{initial_condition_Fp} becomes an initial condition for such quantities.

\section{Linearization and final form of the governing equations}
\label{sec_lin}

We now derive the linearised form of the governing equations under the
assumption of infinitesimal deformations and infinitesimal remodelling
stretches, while retaining the finite character of the rotations
$\boldsymbol{R}_{\rm p}$ and $\boldsymbol{Q}_{\rm p}$, which are not
linearised in the present work. 

\paragraph{Kinematics and stress.} By introducing a smallness positive parameter $\eta\ll 1$, the deformation gradient tensor $\boldsymbol{F}$ reads
\begin{linenomath}
\begin{align}
\boldsymbol{F}
&=
\boldsymbol I+\eta\,\mathrm{Grad}\boldsymbol{u}+o(\eta),
\label{eq_F_lin}
\end{align}
\end{linenomath}
where $\boldsymbol{u}:\mathscr{B}\times\mathscr{I}\rightarrow\mathscr{S}$ is the displacement vector field. Analogously, the tensor $\boldsymbol{D}_{\rm p}$, collecting the principal
remodelling stretches, is expanded as
\begin{linenomath}
\begin{align}
\boldsymbol{D}_{\rm p}
&=
\boldsymbol I+\eta\,\boldsymbol\Lambda_{\rm p}+o(\eta),
\qquad
\boldsymbol{\Lambda}_{\rm p}
=
\mathrm{diag}\{p_{1},p_{2},-p_{1}-p_{2}\},
\label{eq_Dp}
\end{align}
\end{linenomath}
where we have used the first-order expansions
$\mathrm e^{\ell_i}=1+\eta p_i+o(\eta)$, for $i=1,2$, and
$\mathrm e^{-(\ell_1+\ell_2)}=1-\eta(p_1+p_2)+o(\eta)$, where $p_{i}:\mathscr{B}\times\mathscr{I}\rightarrow\mathbb{R}$ is the first-order amplitudes of the logarithmic remodelling stretches $\ell_{i}$. Moreover, $\mathrm{tr}\boldsymbol\Lambda_{\rm p}=0$, consistently with
the isochoric character of $\boldsymbol D_{\rm p}$. Substitution of Equation~\eqref{eq_Dp} into the polar representation of
$\boldsymbol F_{\rm p}$ of Equation~\eqref{eq_polar_dec} yields
\begin{linenomath}
\begin{align}
\boldsymbol{F}_{\rm p}
&=
\boldsymbol{R}_{\rm p}
+\eta\,\boldsymbol{R}_{\rm p}\boldsymbol{Q}_{\rm p}
\boldsymbol{\Lambda}_{\rm p}\boldsymbol{Q}_{\rm p}^{\rm T}
+o(\eta),
&
\boldsymbol{F}_{\rm p}^{-1}
&=
\boldsymbol{R}_{\rm p}^{\rm T}
-\eta\,\boldsymbol{Q}_{\rm p}\boldsymbol{\Lambda}_{\rm p}
\boldsymbol{Q}_{\rm p}^{\rm T}\boldsymbol{R}_{\rm p}^{\rm T}
+o(\eta).
\label{eq_Fp_lin}
\end{align}
\end{linenomath}
Inserting Equations~\eqref{eq_F_lin} and~\eqref{eq_Fp_lin} into
Equation~\eqref{def_BKL_decomp}, we obtain $\boldsymbol F_{\rm e}$ as
\begin{linenomath}
\begin{align}
\boldsymbol F_{\rm e}
&=
\left[
\boldsymbol I+\eta\left(
\mathrm{Grad}\boldsymbol{u}
-\boldsymbol Q_{\rm p}\boldsymbol\Lambda_{\rm p}
\boldsymbol Q_{\rm p}^{\rm T}
\right)
\right]
\boldsymbol{R}_{\rm p}^{\rm T}
+o(\eta).
\label{eq_Fe_lin}
\end{align}
\end{linenomath}
Accordingly, the elastic Green--Lagrange distortion tensor in
Equation~\eqref{eq_Ee} takes the form
\begin{linenomath}
\begin{align}
\boldsymbol E_{\rm e}
=
\eta\,\boldsymbol{R}_{\rm p}
\left[
\boldsymbol{\varepsilon}
-\boldsymbol Q_{\rm p}\boldsymbol\Lambda_{\rm p}
\boldsymbol Q_{\rm p}^{\rm T}
\right]
\boldsymbol{R}_{\rm p}^{\rm T}
+o(\eta), \quad \quad \boldsymbol{\varepsilon}
:=
\mathrm{sym}(\mathrm{Grad}\boldsymbol{u}),
\label{eq_Ee_lin}
\end{align}
\end{linenomath}
where $\boldsymbol{\varepsilon}$ is the infinitesimal strain tensor. Finally, the Mandel stress tensor $\boldsymbol{\Sigma}$ can be linearised as
\begin{linenomath}
\begin{align}
\boldsymbol{\Sigma}
&=
\eta\left[
\lambda\mathrm{tr}(\boldsymbol{\varepsilon})\boldsymbol I
+2\mu
\left(
\boldsymbol{\varepsilon}
-\boldsymbol Q_{\rm p}\boldsymbol\Lambda_{\rm p}
\boldsymbol Q_{\rm p}^{\rm T}
\right)
\right]
+o(\eta).
\label{eq_Sigma_lin}
\end{align}
\end{linenomath}
We remark that, at first order in $\eta$, the finite rotation tensor
$\boldsymbol{R}_{\rm p}$ does not appear explicitly in $\boldsymbol \Sigma$,
since its contributions cancel under the pull-back from the natural
state to the reference configuration. In particular,
$\boldsymbol{\Sigma}$ is symmetric at first order in $\eta$ although this property does not
hold in general in the finite-strain theory. Moreover, since
$\boldsymbol Q_{\rm p}\boldsymbol\Lambda_{\rm p}
\boldsymbol Q_{\rm p}^{\rm T}$ is traceless, the deviatoric part of the
Mandel stress tensor is
\begin{linenomath}
\begin{align}
\mathrm{dev}\boldsymbol\Sigma
&=
2\eta\mu\left[
\mathrm{dev}(\boldsymbol{\varepsilon})
-\boldsymbol Q_{\rm p}\boldsymbol\Lambda_{\rm p}
\boldsymbol Q_{\rm p}^{\rm T}
\right]
+o(\eta).
\label{eq_mandel_dev}
\end{align}
\end{linenomath}

\paragraph{Linearization of the balance of linear momentum.} The linearization of the BVP in Equations \eqref{eq_balance_linear_mom}, \eqref{BVP_tractions} and \eqref{BVP_Dirichlet} is performed by recalling that, at the first order of $\eta$, the first Piola--Kirchhoff stress tensor is 
\begin{linenomath}
\begin{align}
\boldsymbol{P}=
\eta\left[
\lambda\mathrm{tr}(\boldsymbol{\varepsilon})\boldsymbol I
+2\mu
\left(
\boldsymbol{\varepsilon}
-\boldsymbol Q_{\rm p}\boldsymbol\Lambda_{\rm p}
\boldsymbol Q_{\rm p}^{\rm T}
\right)
\right]
+o(\eta),
\label{eq_P_lin}
\end{align}
\end{linenomath}
so that
\begin{linenomath}
\begin{subequations}
\begin{align}
&\mathrm{Div}\left[
\lambda\mathrm{tr}(\boldsymbol{\varepsilon})\boldsymbol I
+2\mu
\left(
\boldsymbol{\varepsilon}
-\boldsymbol Q_{\rm p}\boldsymbol\Lambda_{\rm p}
\boldsymbol Q_{\rm p}^{\rm T}
\right)
\right]=\boldsymbol{f}_{\rm lin}, && \textrm{in\,} \mathscr{B}\times \mathscr{I},
\label{eq_balance_linear_mom2}\\
&\left[
\lambda\mathrm{tr}(\boldsymbol{\varepsilon})\boldsymbol I
+2\mu
\left(
\boldsymbol{\varepsilon}
-\boldsymbol Q_{\rm p}\boldsymbol\Lambda_{\rm p}
\boldsymbol Q_{\rm p}^{\rm T}
\right)
\right]\boldsymbol{N}=\bar{\boldsymbol{\tau}}, && \textrm{on\,} \mathscr{\partial_{\mathrm{N}}B} \times \mathscr{I},
\label{BVP_tractions2}\\
&\boldsymbol{u}=\boldsymbol{u}_{\mathrm{b}}, && \textrm{on\,} \mathscr{\partial_{\mathrm{D}}B}\times \mathscr{I}, 
\label{BVP_Dirichlet2}
\end{align}
\end{subequations}
\end{linenomath}
where $\boldsymbol{u}_{\rm b}$ is the value of $\boldsymbol{u}$ prescribed on the Dirichlet boundary $\mathscr{\partial_{\mathrm{D}}B}$ of $\mathscr{B}$, $\bar{\boldsymbol{\tau}}$ the tractions applied on the Neumann boundary $\mathscr{\partial_{\mathrm{N}}B}$ of $\mathscr{B}$ and $\boldsymbol{f}_{\rm lin}$ is the linearised body force $\boldsymbol{f}$. For consistency with the linearization setting, the prescribed boundary data are assumed to scale as 
\begin{linenomath} 
\begin{align} 
\boldsymbol\tau &= \eta\bar{\boldsymbol\tau}+o(\eta), & \chi_{\mathrm b}(X,t) &= X+\eta\boldsymbol u_{\mathrm b}(X,t)+o(\eta). 
\label{eq_boundary_data_scaling} 
\end{align} 
\end{linenomath}

\paragraph{Linearization of the remodelling law in
Equation \eqref{rem_law_1}.} We first
expand $\boldsymbol L_{\rm p}$ with respect to the remodelling stretches up
to first order in $\eta$, obtaining
\begin{linenomath}
\begin{align}
\boldsymbol L_{\rm p}
&=
\boldsymbol R_{\rm p}^{\rm T}\dot{\boldsymbol R}_{\rm p}
+\eta\left\{\dot{\overline{\boldsymbol Q_{\rm p}\boldsymbol\Lambda_{\rm p}
\boldsymbol Q_{\rm p}^{\rm T}}}
+
\left[
\boldsymbol R_{\rm p}^{\rm T}\dot{\boldsymbol R}_{\rm p},
\boldsymbol Q_{\rm p}\boldsymbol\Lambda_{\rm p}
\boldsymbol Q_{\rm p}^{\rm T}
\right]
\right\}
+o(\eta),
\label{eq_Lp_lin}
\end{align}
\end{linenomath}
where the commutator of two second-order tensors is defined by
\begin{linenomath}
\begin{align}
[\boldsymbol A|\boldsymbol B]
&:=
\boldsymbol A\boldsymbol B-\boldsymbol B\boldsymbol A.
\label{eq_commutator}
\end{align}
\end{linenomath}
The linearization of the right-hand side of the remodelling law in
Equation~\eqref{rem_law_1} then yields
\begin{linenomath}
\begin{align}
b&\left(
\boldsymbol{C}\boldsymbol L_{\rm p}\boldsymbol{C}^{-1}
+\boldsymbol L_{\rm p}^{\rm T}
\right)
+c\left(
\boldsymbol{C}\boldsymbol L_{\rm p}\boldsymbol{C}^{-1}
-\boldsymbol L_{\rm p}^{\rm T}
\right)=
\nonumber\\
&
2c\boldsymbol{R}_{\rm p}^{\rm T}\dot{\boldsymbol{R}}_{\rm p}
+\eta\left\{
2b\left[
\dot{\overline{\boldsymbol Q_{\rm p}\boldsymbol\Lambda_{\rm p}
\boldsymbol Q_{\rm p}^{\rm T}}}
+
\left[
\boldsymbol{R}_{\rm p}^{\rm T}\dot{\boldsymbol{R}}_{\rm p}|
\boldsymbol Q_{\rm p}\boldsymbol{\Lambda}_{\rm p}
\boldsymbol Q_{\rm p}^{\rm T}
\right]
\right]
+2(b+c)
\left[
\boldsymbol{\varepsilon}|
\boldsymbol{R}_{\rm p}^{\rm T}\dot{\boldsymbol{R}}_{\rm p}
\right]
\right\}
+o(\eta).
\label{eq_rem_lin}
\end{align}
\end{linenomath}
The leading term on the right-hand side of Equation~\eqref{eq_rem_lin} is of
order $\eta^0$. It must vanish because Equation~\eqref{eq_rem_lin} is
equated, through the remodelling law, to
Equation~\eqref{eq_mandel_dev}, whose zeroth-order term is null. Equivalently,
Equation~\eqref{eq_mandel_dev} is the first-order linearization of the
right-hand side of Equation~\eqref{rem_law_1}, whereas
Equation~\eqref{eq_rem_lin} is the corresponding first-order linearization
of the right-hand side of Equation~\eqref{rem_law_1}. Thus, for
$c\neq0$, the zeroth-order balance gives
\begin{linenomath}
\begin{align}
\boldsymbol R_{\rm p}^{\rm T}\dot{\boldsymbol R}_{\rm p}
=
\boldsymbol 0, \quad \mbox{ and equivalently, } \quad \dot{\boldsymbol R}_{\rm p}
=
\boldsymbol 0.
\label{eq_Rp_spin_zero}
\end{align}
\end{linenomath}
Upon prescribing the initial condition
\begin{linenomath}
\begin{align}
\boldsymbol R_{\rm p}(X,0)
&=
\boldsymbol I, \qquad X\in\mathscr{B},
\label{eq_Rp_initial}
\end{align}
\end{linenomath}
one obtains
\begin{linenomath}
\begin{align}
\boldsymbol R_{\rm p}(X,t)
&=
\boldsymbol I, \qquad (X,t)\in\mathscr{B}\times\mathscr{I}.
\label{eq_Rp_identity}
\end{align}
\end{linenomath}
\begin{remark}[On the choice of the remodelling rotation] It is important to distinguish between  $\dot{\boldsymbol R}_{\rm p}=\boldsymbol 0$ and $\boldsymbol R_{\rm p}=\boldsymbol I$. The former is an evolution assumption and implies only that
\begin{align}
\boldsymbol R_{\rm p}(X,t)
=
\boldsymbol R_{\rm p}^{0}(X),
\qquad (X,t)\in\mathscr{B}\times\mathscr{I},
\end{align}
whereas the latter requires the additional initial condition
\begin{align}
\boldsymbol R_{\rm p}(X,0)=\boldsymbol R_{\rm p}^{0}(X)=\boldsymbol I, \qquad X\in\mathscr{B}.
\end{align}
Consequently, the identity
$\boldsymbol R_{\rm p}(X,t)=\boldsymbol I$, $(X,t)\in\mathscr{B}\times\mathscr{I}$, follows only from imposing both conditions. To motivate this choice, consider the reference configuration,
for which the total deformation and the plastic stretch satisfy
\begin{align}
\boldsymbol F(X,0)=\boldsymbol I,
\qquad
\boldsymbol U_{\rm p}(X,0)=\boldsymbol I, \qquad X\in\mathscr{B}.
\end{align}
If $\dot{\boldsymbol R}_{\rm p}=\boldsymbol 0$, then
$\boldsymbol F_{\rm p}(X,0)=\boldsymbol R_{\rm p}(X,0)$ and hence
\begin{align}
\boldsymbol F_{\rm e}(X,0)
=
\boldsymbol F(X,0)
\boldsymbol F_{\rm p}^{-1}(X,0)
=
\boldsymbol R_{\rm p}^{\mathrm T}(X,0).
\end{align}
Thus, the initial elastic distortion is purely rotational and this rotation produces no elastic strain. Therefore, an initial remodelling rotation is locally
invisible to the isotropic constitutive response. Choosing
$\boldsymbol R_{\rm p}^{0}(X)=\boldsymbol I$, with $X\in\mathscr{B}$, does not remove any initial
elastic strain or stress but simply selects an initially unrotated
representative of the natural state. We remark that this restriction is exact and amounts to
removing a time-independent rotational degree of freedom that is not
activated by the local isotropic behaviour of $\mathscr{B}$. A spatially uniform initial rotation could alternatively be removed by a single global change of frame. In contrast, a spatially varying field
$\boldsymbol R_{\rm p}^{0}(X)=\boldsymbol R_{\rm p}(X,0)$ cannot generally be eliminated in this
way. Although it remains locally stress-free under isotropy, its spatial variation may generate incompatibility and
become relevant in theories involving material gradients, non-local
effects, or energetic terms depending on
$\operatorname{Grad}\boldsymbol R_{\rm p}$.
\end{remark}
\noindent
After imposing $\boldsymbol{R}_{\rm p}=\boldsymbol{I}$ and in light of Equation \eqref{eq_mandel_dev}, Equation \eqref{eq_rem_lin} becomes
\begin{linenomath}
\begin{align}
\dot{\overline{\boldsymbol Q_{\rm p}\boldsymbol\Lambda_{\rm p}\boldsymbol Q_{\rm p}^{\rm T}}}=
\frac{\mu}{b}
\left[
\mathrm{dev}(\boldsymbol{\varepsilon})
-\boldsymbol Q_{\rm p}\boldsymbol\Lambda_{\rm p}
\boldsymbol Q_{\rm p}^{\rm T}
\right].
\label{eq_final_linearized_evolution2}
\end{align}
\end{linenomath}
The right-hand side of Equation \eqref{eq_final_linearized_evolution2} defines a linear relaxation towards the deviatoric part of the infinitesimal strain, expressed in the basis of the principal remodelling directions with characteristic remodelling time $t_{\mathrm r}:=b/\mu$. Hence, larger values of $t_{\mathrm r}$ describe a slower remodelling response, whereas smaller values correspond to a faster adaptation of the natural state. Moreover, although the coefficient $c$ does not appear explicitly in Equation~\eqref{eq_final_linearized_evolution2}, it has not been set equal to zero. Its role is already exhausted at order $\eta^{0}$: because $c>0$, the zeroth-order part of the constitutive equation enforces $\boldsymbol R_{\rm p}^{\rm T}\dot{\boldsymbol R}_{\rm p}=\boldsymbol0$. Once this spin has vanished, the first-order remodelling velocity is symmetric and the skew-like resistance governed by $c$ is no longer activated and the linearised relaxation is therefore controlled only by $b$. We next introduce the spin associated with $\dot{\boldsymbol{Q}}_{\rm p}$
\begin{linenomath}
\begin{align}
\boldsymbol\Omega
&:=
\boldsymbol Q_{\rm p}^{\rm T}\dot{\boldsymbol Q}_{\rm p},
\qquad
\boldsymbol\Omega^{\rm T}=-\boldsymbol\Omega,
\label{eq_Omega_def}
\end{align}
\end{linenomath}
and with the following kinematic identity holding true
\begin{linenomath}
\begin{align}
\dot{\overline{\boldsymbol Q_{\rm p}\boldsymbol\Lambda_{\rm p}\boldsymbol Q_{\rm p}^{\rm T}}}=\boldsymbol Q_{\rm p}(\dot{\boldsymbol{\Lambda}}_{\rm p}
+\boldsymbol\Omega\boldsymbol\Lambda_{\rm p}
-\boldsymbol\Lambda_{\rm p}\boldsymbol\Omega)\boldsymbol Q_{\rm p}^{\rm T}.
\label{eq_Omega_def2}
\end{align}
\end{linenomath}
Since $\dot{\boldsymbol\Lambda}_{\rm p}$ is symmetric, $\boldsymbol\Omega$ is skew-symmetric, and $\boldsymbol\Lambda_{\rm p}$ is symmetric, the commutator $\boldsymbol\Omega\boldsymbol\Lambda_{\rm p} -\boldsymbol\Lambda_{\rm p}\boldsymbol\Omega$ is symmetric. Consequently, the first-order remodelling velocity is symmetric after setting $\boldsymbol R_{\rm p}=\boldsymbol I$. Accordingly, using the deviatoric part of the linearised Mandel stress, and recalling that it is symmetric too, we obtain the final form of the linearised remodelling law
\begin{linenomath}
\begin{align}
\dot{\boldsymbol{\Lambda}}_{\rm p}
+\boldsymbol\Omega\boldsymbol\Lambda_{\rm p}
-\boldsymbol\Lambda_{\rm p}\boldsymbol\Omega
&=
\frac{\mu}{b}
\left[
\boldsymbol Q_{\rm p}^{\rm T}
\mathrm{dev}(\boldsymbol{\varepsilon})
\boldsymbol Q_{\rm p}
-\boldsymbol\Lambda_{\rm p}
\right].
\label{eq_final_linearized_evolution}
\end{align}
\end{linenomath}
Equation \eqref{eq_final_linearized_evolution} describes the coupled evolution of the principal remodelling stretches and of their associated principal directions. The term $\dot{\boldsymbol\Lambda}_{\rm p}$ accounts for the evolution of the principal remodelling stretches, whereas the commutator $\boldsymbol\Omega\boldsymbol\Lambda_{\rm p} -\boldsymbol\Lambda_{\rm p}\boldsymbol\Omega$ arises from the rotation of the principal remodelling frame and provides the nonlinear coupling between stretches and orientations. Consequently, even within the infinitesimal-strain regime, the remodelling dynamics remains nonlinear through the interaction between the principal stretches and the finite rotations of their corresponding principal directions. Finally, Equation~\eqref{eq_final_linearized_evolution} will constitute the starting point for the controllability analysis developed in the subsequent sections.

\begin{remark}[Initial conditions in the linearised setting]
The initial condition in Equation~\eqref{initial_condition_Fp} must be consistent with the assumed smallness of the remodelling stretches. After choosing $\boldsymbol R_{\rm p}^{0}=\boldsymbol I$, the admissible initial data have the expansion
\begin{linenomath}
\begin{align}
\Fp^{0}(X)
=
\boldsymbol I
+\eta\,
\boldsymbol Q_{\rm p}^{0}(X)
\boldsymbol\Lambda_{\rm p}^{0}(X)
\bigl(\boldsymbol Q_{\rm p}^{0}(X)\bigr)^{\rm T}
+o(\eta),
\qquad X\in\mathscr B,
\label{eq_Fp0_linearized}
\end{align}
\end{linenomath}
where
\begin{linenomath}
\begin{subequations}
\begin{align}
\boldsymbol\Lambda_{\rm p}^{0}(X)
&=\boldsymbol\Lambda_{\rm p}(X,0)=\mathrm{diag}\{p_1^0(X),p_2^0(X),-p_1^0(X)-p_2^0(X)\},
\\
\boldsymbol Q_{\rm p}^{0}(X)
&=\boldsymbol Q_{\rm p}(X,0).
\label{eq_linearized_initial_data}
\end{align}
\end{subequations}
\end{linenomath}
with, $p_1(X,0)=p_1^0(X)$ and $p_2(X,0)=p_2^0(X)$. The initial orientation $\boldsymbol Q_{\rm p}^{0}$ constitutes an independent kinematic datum only when the principal values of $\boldsymbol\Lambda_{\rm p}^{0}$ are pairwise distinct. If two principal values coincide, rotations within the corresponding eigenspace become immaterial. In particular, if $\boldsymbol\Lambda_{\rm p}^{0}=\boldsymbol 0$, the initial Euler angles are completely arbitrary and carry no independent physical information.
\end{remark}

\section{State-space representation of the remodelling law}
\label{sec_state}
The linearised remodelling law derived in Section~\ref{sec_lin} is now reformulated as a finite-dimensional nonlinear state-space system for the principal remodelling stretches $p_{1}$ and $p_{2}$ and the Euler angles $\alpha, \beta$ and $\gamma$ describing the orientation of the principal remodelling directions. By inserting Equation \eqref{eq_Qp_explicit} into Equation  \eqref{eq_Omega_def}, and for the adopted Euler-angle convention, we obtain the explicit representation of the spin tensor $\boldsymbol{\Omega}$ as 
\begin{linenomath}
\begin{align}
\boldsymbol\Omega
=
\begin{pmatrix}
0
&
-\dot{\alpha}\cos\beta-\dot{\gamma}
&
\dot{\alpha}\sin\beta\sin\gamma+\dot{\beta}\cos\gamma
\\
\dot{\alpha}\cos\beta+\dot{\gamma}
&
0
&
\dot{\alpha}\sin\beta\cos\gamma-\dot{\beta}\sin\gamma
\\
-\dot{\alpha}\sin\beta\sin\gamma-\dot{\beta}\cos\gamma
&
-\dot{\alpha}\sin\beta\cos\gamma+\dot{\beta}\sin\gamma
&
0
\end{pmatrix}.
\label{eq_Omega_Euler_explicit}
\end{align}
\end{linenomath}

\paragraph{State arrays.}
We introduce the remodelling state array $z=\begin{bmatrix}
p_1&
p_2&
\alpha&
\beta&
\gamma
\end{bmatrix}^{\rm T}\in\mathbb{R}^{5}$, which collects the internal degrees of freedom governing the remodelling process and, since $\mathrm{dev}(\boldsymbol\varepsilon)$ is symmetric and traceless, it is completely determined by five independent components and, therefore, we introduce the strain state array $
\varepsilon_{\rm dev} = \begin{bmatrix}
\mathrm{dev}(\boldsymbol\varepsilon)_{11}&\mathrm{dev}(\boldsymbol\varepsilon)_{22}&\mathrm{dev}(\boldsymbol\varepsilon)_{12}&\mathrm{dev}(\boldsymbol\varepsilon)_{13}&\mathrm{dev}(\boldsymbol\varepsilon)_{23}
\end{bmatrix}^{\rm T}\in\mathbb{R}^{5}$.

\paragraph{Reduced evolution system.}
Equation \eqref{eq_state_space_remodelling} provides an implicit nonlinear state-space representation of the remodelling dynamics. By selecting the components $(11)$, $(22)$, $(12)$, $(13)$ and $(23)$
of Equation~\eqref{eq_final_linearized_evolution}, the remodelling law can
be written as an implicit first-order state-space system
\begin{linenomath}
\begin{align}
\boldsymbol M(z)\dot{\boldsymbol z}
=
\boldsymbol b(z,\varepsilon_{\rm dev}).
\label{eq_state_space_remodelling}
\end{align}
\end{linenomath}
The state matrix $\boldsymbol M(z)$ multiplying $\dot{\boldsymbol{z}}$ in Equation \eqref{eq_state_space_remodelling} is represented by
\begin{linenomath}
\begin{align}
\boldsymbol M( z)
=
\begin{pmatrix}
1&0&0&0&0
\\
0&1&0&0&0
\\
0&0&(p_1-p_2)\cos\beta&0&p_1-p_2
\\
0&0&-(2p_1+p_2)\sin\beta\sin\gamma
&-(2p_1+p_2)\cos\gamma&0
\\
0&0&-(p_1+2p_2)\sin\beta\cos\gamma
&(p_1+2p_2)\sin\gamma&0
\end{pmatrix},
\label{eq_A_state}
\end{align}
\end{linenomath}
while the array $\boldsymbol{b}(z,\varepsilon_{\rm dev})$ is given by
\begin{linenomath}
\begin{align}
\boldsymbol b( z,\varepsilon_{\rm dev})
=
\frac{\mu}{b}
\begin{pmatrix}
\widehat{\varepsilon}_{11}-p_1
&
\widehat{\varepsilon}_{22}-p_2
&
\widehat{\varepsilon}_{12}
&
\widehat{\varepsilon}_{13}
&
\widehat{\varepsilon}_{23}
\end{pmatrix}^{\rm T}, \quad \widehat{\boldsymbol\varepsilon}
:=
\Qp^{\rm T}\,
\mathrm{dev}(\boldsymbol{\varepsilon})\,
\Qp.
\label{eq_rotated_strain_state}
\end{align}
\end{linenomath}
The entries of $\boldsymbol b$ are completely determined by
Equation \eqref{eq_rotated_strain_state}. Their
fully expanded expressions are algebraically lengthy and provide no
additional insight into the structure of the evolution law and,
therefore, are not reported explicitly. 

\paragraph{Invertibility of the state matrix and drift term.} Direct calculation
yields
\begin{linenomath}
\begin{align}
\det\boldsymbol M(z)
=
-(p_1-p_2)(2p_1+p_2)(p_1+2p_2)\sin\beta.
\label{eq_det_A_state}
\end{align}
\end{linenomath}
Recalling that the principal values of $\boldsymbol\Lambda_{\rm p}$ are
\begin{linenomath}
\begin{align}
\Lambda_1=p_1,
\qquad
\Lambda_2=p_2,
\qquad
\Lambda_3=-p_1-p_2,
\label{eq_principal_Lambda_values}
\end{align}
\end{linenomath}
the stretch-dependent factors appearing in
Equation~\eqref{eq_det_A_state} can be expressed as
\begin{linenomath}
\begin{align}
p_1-p_2
=
\Lambda_1-\Lambda_2, \qquad 2p_1+p_2
=
\Lambda_1-\Lambda_3, \qquad
p_1+2p_2
=
\Lambda_2-\Lambda_3.
\label{eq_invertibilityA}
\end{align}
\end{linenomath}
The vanishing of the determinant identifies two distinct sources of singularity: (i) coincident principal remodelling stretches and (ii) the gimbal-lock singularity associated with the Euler-angle parametrization. Hence, $\boldsymbol M$ is invertible if the principal remodelling stretches must be pairwise distinct and $\beta\neq 0$ and, on the open subset defined by
Equation \eqref{eq_invertibilityA}, there exists the
inverse matrix $\boldsymbol M^{-1}(z)$ reading
\begin{linenomath}
\begin{align}
\boldsymbol M^{-1}(z)
=
\begin{pmatrix}
1&0&0&0&0
\\
0&1&0&0&0
\\
0&0&0&
-\dfrac{\sin\gamma}{(2p_1+p_2)\sin\beta}
&
-\dfrac{\cos\gamma}{(p_1+2p_2)\sin\beta}
\\[3mm]
0&0&0&
-\dfrac{\cos\gamma}{2p_1+p_2}
&
\dfrac{\sin\gamma}{p_1+2p_2}
\\[3mm]
0&0&
\dfrac{1}{p_1-p_2}
&
\dfrac{\cos\beta\sin\gamma}
{(2p_1+p_2)\sin\beta}
&
\dfrac{\cos\beta\cos\gamma}
{(p_1+2p_2)\sin\beta}
\end{pmatrix}.
\label{eq_A_inverse_state}
\end{align}
\end{linenomath}
This leads to the explicit coordinate form of Equation \eqref{eq_state_space_remodelling}, namely
\begin{linenomath}
\begin{align}
\dot{\boldsymbol z}
=\boldsymbol{f}_{0}(z,\varepsilon_{\rm dev}):=
\boldsymbol M^{-1}(z)\boldsymbol b(z,\varepsilon_{\rm dev}).
\label{eq_explicit}
\end{align}
\end{linenomath}
In the terminology of Geometric Control Theory, the vector field $\boldsymbol f_0$ is the drift of the system, namely the uncontrolled dynamics of the remodelling law in the absence of external inputs.

\paragraph{Equilibrium points.} We provide a general result regarding the existence of equilibria for Equation \eqref{eq_explicit} and, equivalently, for Equation \eqref{eq_state_space_remodelling}. The following theorem holds true.

\begin{theorem}
\label{thm_equilibrium}
The equilibrium configurations $(\boldsymbol Q_{\rm p},\boldsymbol\Lambda_{\rm p})$ of the remodelling law \eqref{eq_final_linearized_evolution} (equivalently of \eqref{eq_explicit}) coincide with the diagonal matrix of eigenvalues of $\mathrm{dev}(\boldsymbol\varepsilon)$ and $\boldsymbol Q_{\rm p}$ must align the principal remodelling directions with the corresponding eigenvectors.
\end{theorem}

\begin{proof}
An equilibrium of Equation \eqref{eq_final_linearized_evolution} is a pair $(\boldsymbol Q_{\rm p},\boldsymbol\Lambda_{\rm p})$ for which 
\begin{linenomath}
\begin{align}
\boldsymbol Q_{\rm p}^{\rm T}\,\mathrm{dev}(\boldsymbol\varepsilon)\,\boldsymbol Q_{\rm p}-\boldsymbol\Lambda_{\rm p}=\boldsymbol 0.
\label{eq_equilibrium_condition}
\end{align}
\end{linenomath}
Since $\mathrm{dev}(\boldsymbol\varepsilon)$ is symmetric, there exists an orthogonal matrix $\boldsymbol{\mathcal Q}$ and a diagonal matrix $\boldsymbol\epsilon=\operatorname{diag}(\epsilon_1,\epsilon_2,\epsilon_3)$, with $\epsilon_1+\epsilon_2+\epsilon_3=0$, such that
\begin{linenomath}
\begin{align}
\mathrm{dev}(\boldsymbol\varepsilon)=\boldsymbol{\mathcal Q}\,\boldsymbol\epsilon\,\boldsymbol{\mathcal Q}^{\rm T}.
\end{align}
\end{linenomath}
Insert this spectral decomposition into Equation \eqref{eq_equilibrium_condition} to obtain
\begin{linenomath}
\begin{align}
\boldsymbol Q_{\rm p}^{\rm T}\,\boldsymbol{\mathcal Q}\,\boldsymbol\epsilon\,\boldsymbol{\mathcal Q}^{\rm T}\,\boldsymbol Q_{\rm p}-\boldsymbol\Lambda_{\rm p}=\boldsymbol 0.
\end{align}
\end{linenomath}
Both sides are symmetric and diagonalizable. The left‑hand side is a diagonal matrix if and only if $\boldsymbol Q_{\rm p}^{\rm T}\,\boldsymbol{\mathcal Q}$ permutes (and possibly changes the sign of) the eigenvector basis of $\boldsymbol\epsilon$. Hence the diagonal entries of \(\boldsymbol\Lambda_{\rm p}\) must be a permutation of the eigenvalues \(\epsilon_i\). Therefore one can choose \(\boldsymbol Q_{\rm p}=\boldsymbol{\mathcal Q}\,\boldsymbol\Pi\), where \(\boldsymbol\Pi\) is a permutation (and sign) matrix that orders the eigenvalues on the diagonal to match the convention for \(\boldsymbol\Lambda_{\rm p}\). With this choice \(\boldsymbol\Lambda_{\rm p}=\boldsymbol\epsilon\) (up to the chosen ordering), and \eqref{eq_equilibrium_condition} is satisfied. This proves that equilibria correspond to taking \(\boldsymbol\Lambda_{\rm p}\) equal to the eigenvalues of \(\mathrm{dev}(\boldsymbol\varepsilon)\) and \(\boldsymbol Q_{\rm p}\) aligning the principal remodelling directions with the corresponding eigenvectors.
\end{proof}

\begin{remark}[Degenerate spectra]
If \(\mathrm{dev}(\boldsymbol\varepsilon)\) has repeated eigenvalues, the equilibrium set is not a discrete set of pairs \((\boldsymbol Q_{\rm p},\boldsymbol\Lambda_{\rm p})\) but contains continuous families: when two (or three) eigenvalues coincide, rotations within the corresponding eigenspace leave \(\boldsymbol\epsilon\) invariant, so \(\boldsymbol Q_{\rm p}\) is determined only up to an orthogonal transformation acting on the repeated eigenspace. In particular, if \(\boldsymbol\epsilon=\boldsymbol0\) then any \(\boldsymbol Q_{\rm p}\) together with \(\boldsymbol\Lambda_{\rm p}=\boldsymbol0\) is an equilibrium. These degenerate configurations are excluded from the controllability analysis developed in the subsequent sections.
\end{remark}

\paragraph{Time-independent strain fields and mechanical equilibrium.} Throughout the present work, the infinitesimal strain tensor is assumed to be prescribed and independent of time. Consequently, the deformation does not evolve in time and acts instead as a fixed mechanical loading entering the remodelling law through its deviatoric part. 

Under this assumption, the balance of linear momentum is not used to determine the time evolution of the strain field. Rather, at each instant of time, it specifies the external body forces and boundary tractions required to maintain the prescribed deformation while the remodelling state evolves. In this sense, remodelling takes place in a mechanically sustained configuration, and the corresponding equilibrium loads may vary in time through their dependence on the evolving remodelling tensor and on the associated stress fields. 

Accordingly, the controllability analysis developed below concerns the internal remodelling variables alone, while the strain field is regarded as a given driving quantity and the linear momentum balance provides the mechanical actions necessary to support the prescribed deformation throughout the remodelling process.

A similar viewpoint was already adopted in \cite{DiStefano2025a} in the context of volumetric growth. Although this modelling assumption was not explicitly emphasised there, the growth dynamics was likewise analysed under prescribed mechanical conditions, with the balance equations intended to determine the loads required to sustain the assigned state while the internal variables evolved. 

\section{Brief revision of Geometric Control Theory}
\label{sec_control}
The remodelling law in Equation~\eqref{eq_final_linearized_evolution}, or equivalently its implicit and explicit state-space representations in Equations~\eqref{eq_state_space_remodelling} and~\eqref{eq_explicit}, respectively, can be interpreted within the framework of Geometric Control Theory. From this perspective, the remodelling variables define the state of the system, while suitable control functions act as external inputs that may be selected to steer the remodelling state towards a prescribed target or along a desired trajectory. This control-theoretic interpretation is consistent with the local character of the adopted remodelling theory, in which remodelling is described as a time-evolution process with the material point $X\in\mathscr B$ acting as a parameter. Since the theory is of order zero in the remodelling tensor $\boldsymbol F_{\rm p}$, no spatial derivatives of $\boldsymbol F_{\rm p}$ are included among the kinematic descriptors. Consequently, for each fixed material point $X\in\mathscr B$, the remodelling variables satisfy a finite-dimensional system of ordinary differential equations in time. 

Before introducing the control functions, the uncontrolled remodelling dynamics is described by the implicit state-space system in Equation~\eqref{eq_state_space_remodelling}. In this formulation, the infinitesimal strain tensor is not a control input, but a mechanical quantity entering the constitutive remodelling law and, in the coupled problem, is determined by the displacement field through the balance equations, compatibility requirements and boundary conditions. The control functions are introduced separately as additional external actions on the remodelling dynamics that may depend on both the material point and time, so that the controlled model defines a family of ordinary differential systems parametrised by $X\in\mathscr B$. For each fixed material point, the controls act as time-dependent inputs to the corresponding local remodelling system. Therefore, provided that suitable regularity with respect to $X$ is assumed, the standard finite-dimensional tools of Geometric Control Theory can be applied pointwise in space. In particular, the controllability properties of the associated ordinary differential control system may be investigated separately at each material point. Within this setting, a distinction must be made between spatially distributed and spatially uniform controls. Spatially distributed controls may vary from one material point to another, whereas spatially uniform controls prescribe the same time-dependent control history throughout $\mathscr B$. The latter condition does not follow from the zero-order character of the remodelling theory and must therefore be introduced as an additional modelling assumption on the class of admissible controls.

In this section, we recall the notation and the main concepts of Geometric Control Theory required for the subsequent analysis and specialise them to the controlled remodelling dynamics.

\begin{definition}[Control system \cite{Coron2007}]
\label{def_contolled_problem}
A control system is a system of first-order, ordinary differential equations written, on its normal form, as
\begin{equation}
\label{eq: control system}
    \dot{\boldsymbol{z}} = \boldsymbol{\mathcal{F}}(z,w),
\end{equation}
where $z \in\mathcal{Z}\subset \mathbb{R}^{n}$ collects state variables of the system at hand, \( w : \mathscr{I} \rightarrow \mathcal{U} \subset \mathbb{R}^{m} \), is a measurable function belonging to $\mathcal{U}$, being the latter an open subset of $\mathbb{R}^{m}$ referred to as the set of admissible controls, while $\boldsymbol{\mathcal{F}}:\mathcal{Z}\times\mathcal{U}\rightarrow \mathbb{R}^{n}$ is a (sufficiently regular) vector field.
\end{definition}
\noindent
We say that the control system in Equation \eqref{eq: control system} is non-linear if the vector field $\boldsymbol{\mathcal{F}}$ is non-linear as a function of $z$ and $w$ and, hence, defines a non-linear system of differential equations. Moreover, $m\leq n$ and if $m<n$ strictly, we say that the control system is under-actuated, since the number of controls is less than the number of state variables. Moreover, for further use, we define the open subset $\mathcal{O}:=\mathcal{Z}\times\mathcal{U}\subset \mathbb{R}^{n}\times\mathbb{R}^{m}$.

\begin{definition}[Equilibrium point \cite{Coron2007}]
\label{equilibrium_pt}
An equilibrium point of the control system in Equation \eqref{eq: control system}
is a pair \((z_{\mathrm{eq}}, w_{\mathrm{eq}}) \in \mathcal{O}\) such that  
\begin{equation}
\boldsymbol{\mathcal{F}}(z_{\mathrm{eq}}, w_{\mathrm{eq}}) = \boldsymbol{0}. 
\end{equation}
\end{definition}
\noindent
A crucial notion in Geometric Control Theory is \textit{controllability}, which pertains to the ability of a system to be driven from an arbitrary initial state to any specified final state within a finite time interval, under the influence of control inputs belonging to a certain functional space.

\begin{definition}[Controllability \cite{Coron2007}]
\label{def: controllbility}
The control system in Equation \eqref{eq: control system} is said to be controllable if, for any $z_{\rm in},\,z_{\rm fin}\in\mathcal{Z}$, there exists a control function \( w : \mathscr{I} \rightarrow \mathcal{U} \) such that the solution to \eqref{eq: control system} with that control satisfies $z(t_{\rm in})=z_{\rm in}$ and $z(t_{\rm fin})=z_{\rm fin}$.
\end{definition}

\noindent
In other words, controllability means that a dynamic system as the one in Equation \eqref{eq: control system} can be {\it driven} or {\it controlled} to pass from any initial state  $z_{\rm in}$ to any target state $z_{\rm fin}$ in the bounded time interval $\mathscr{I}$. The control \( w\) is the external input or influence required to accomplish this. The notion of under-actuation refers to the case $m<n$, namely when the number of control inputs is smaller than the number of state variables.

\medskip
\noindent
For further use, we introduce a weaker form of controllability, relying on Definition \ref{equilibrium_pt} of equilibrium points. In this sense, we speak of \emph{small-time local controllability} (STLC), which refers to controllability with controls $w$ ``near'' \( w_{\mathrm{eq}} \) over a short time horizon.

\begin{definition}[Small Time Local Controllability \cite{Coron2007}]
\label{STLC_def}
Let \((z_{\mathrm{eq}}, w_{\mathrm{eq}}) \in \mathcal{O}\) be an equilibrium of the control system in Equation \eqref{eq: control system}. Then, the system is small-time locally controllable (STLC) at the equilibrium \((z_{\mathrm{eq}}, w_{\mathrm{eq}})\) if, for every \(\rho > 0\), there exists \(r > 0\) such that, for every \( z_{\rm in} \in B_r(z_{\mathrm{eq}}) := \{z \in \mathbb{R}^n : |z - z_{\mathrm{eq}}| < r\} \) and for every \(z_{\rm fin} \in B_r(z_{\mathrm{eq}})\), there exists a measurable function \( w : [t_{\rm in}, t_{\rm in}+\rho] \to \mathcal{U} \) such that  
\begin{linenomath}
\begin{align}
|w(t) - w_{\mathrm{eq}}| \leq \rho, \; \forall t \in [t_{\rm in}, t_{\rm in}+\rho],
\end{align}
\end{linenomath}
and
\begin{linenomath}
\begin{align}
\begin{cases}
\dot{\boldsymbol{z}} = \boldsymbol{\mathcal{F}}(z, w),\\
z(t_{\rm in}) = z_{\rm in}, \quad
 z(t_{\rm in}+\rho)=z_{\rm fin}.
\end{cases}
\end{align}
\end{linenomath}
\end{definition}
\noindent
The condition STLC ensures that a dynamic system is controllable within a small neighbourhood of the equilibrium point
and in an arbitrarily short amount of time, weakening the more general concept of global controllability, and  providing great accuracy.

\medskip
\noindent
From now on, let us fix the form of the control system in Equation \eqref{eq: control system} as
\begin{equation}
\label{eq:affine}
\dot{\boldsymbol{z}}=\boldsymbol{\mathcal{F}}(z,w)=\boldsymbol{f}_0(z)+\sum_{i=1}^{m}\boldsymbol{f}_{i}(z)w_i,
\end{equation}
where the vector field $\boldsymbol{\mathcal{F}}$ is non-linear in the state variables $z$ and affine in the control function $w$, being $w=(w_{1}, w_{2},\dots,w_{m})$, with non-zero drift term $\boldsymbol{f}_0$ and with vector fields $\boldsymbol{f}_j:\mathcal{Z}\rightarrow\mathbb{R}^{n}$ being non-linear in $z$, for $i=0,\dots,m$. 

\begin{definition}[Equilibria of affine control systems]
\label{def_equilibria_aff}
From here on, we consider equilibrium points of the form
\begin{linenomath}
\begin{align}
(z_{\rm eq},w_{\rm eq})=\left(\bar{z}_{\rm eq},0\right)
\label{eq_eq_affine}
\end{align}
\end{linenomath}
for the control system of the form of Equation \eqref{eq:affine}, where $\bar{z}_{\rm eq}$ is a zero of the drift term, as in Section \ref{sec_state}, i.e., an equilibrium of the non-controlled remodelling law.  Such a definition is introduced and employed in \cite{DiStefano2025a} (see, in particular, Section 7).
\end{definition}
\noindent
In this context, we need the following definition of Lie algebra.

\begin{definition}[Lie Algebra \cite{Coron2007}]
Consider a control system in the form of Equation \eqref{eq:affine}, an equilibrium point $(z_{\rm eq},w_{\rm eq})=\left(\bar{z}_{\rm eq},0\right)$ of that system, as defined in Equation \eqref{eq_eq_affine}, and the related vector fields $\boldsymbol{f}_0, \ldots, \boldsymbol{f}_m$. The Lie algebra $\mathcal{A}$ formed by the vector fields $\boldsymbol{f}_0, \ldots, \boldsymbol{f}_m$ with respect to $(z_{\rm eq},w_{\rm eq})=\left(\bar{z}_{\rm eq},0\right)$ is defined as
\begin{linenomath}
\begin{align}
\mathcal{A}\left(\bar{z}_{\rm eq}\right)=\left\lbrace\boldsymbol{g}\left(\bar{z}_{\rm eq}\right)\in \mathbb{R}^n \,\, | \,\, \boldsymbol{g} \in \mathrm{Lie}(\boldsymbol{f}_0,\boldsymbol{f}_1,\ldots, \boldsymbol{f}_m)\right\rbrace,
\end{align}
\end{linenomath}
being $\mathrm{Lie}(\boldsymbol{f}_0,\boldsymbol{f}_1,\ldots, \boldsymbol{f}_m)$ the Lie algebra generated by the vector fields $\boldsymbol{f}_0,\boldsymbol{f}_1,\ldots, \boldsymbol{f}_m$.
\end{definition}
\noindent
To prove controllability of systems reported in Equation \eqref{eq:affine}, we require the following property to hold true.
\begin{definition}[Lie Algebra Rank Condition \cite{Coron2007}]
\label{LARC}
The control system in Equation \eqref{eq:affine} satisfies the Lie Algebra Rank Condition (LARC) at an equilibrium point $(z_{\mathrm{eq}},w_{\mathrm{eq}})$ if
\begin{equation}
\mathcal{A}(z_{\mathrm{eq}},w_{\mathrm{eq}}) = \mathbb{R}^n.
\end{equation}
\end{definition}
\begin{proposition}
\label{propositionLARC}
If the vector fields $\boldsymbol{f}_0,\boldsymbol{f}_1,\ldots, \boldsymbol{f}_m$ in Equation \eqref{eq:affine} are analytic and if they do not satisfy the LARC at an equilibrium point, then the control in Equation system \eqref{eq:affine} is not STLC at that point. In other words, for non-linear affine control systems with drift and analytic vector fields $\boldsymbol{f}_0,\boldsymbol{f}_1,\ldots, \boldsymbol{f}_m$, the LARC condition is a necessary condition for STLC.
\end{proposition}

\noindent
The Kalman criterion will be employed repeatedly throughout the paper to analyse the controllability of linearized remodelling systems. In the specific case in which the control system is linear, i.e. $\dot{\boldsymbol{z}}=\boldsymbol{\mathcal{F}}(z,w)=\boldsymbol{A}_{\mathrm{lin}}z+\boldsymbol{B}_{\mathrm{lin}}w$, its controllability can be proven using the Kalman criterion that we state here
\begin{proposition}[Kalman Criterion]
    The  linear control system $\dot{\boldsymbol{z}}=\boldsymbol{A}_{\mathrm{lin}}z+\boldsymbol{B}_{\mathrm{lin}}w$ is
controllable on $[t_{\rm in}, t_{\rm fin}]$ if and only if
\begin{linenomath}
\begin{align}
\mathrm{span} \{\boldsymbol{A}_{\mathrm{lin}}^i\boldsymbol{B}_{\mathrm{lin}}w; w\in  \mathbb{R}^m, i=0,\ldots n-1
\} = \mathbb{R}^n.
\end{align}
\end{linenomath}
\end{proposition}
\noindent
Therefore introducing the \emph{Kalman matrix} $\mathsf{C}_{\mathrm{Kal}}$, defined as
\begin{linenomath}
\begin{align}
\mathsf{C}_{\mathrm{Kal}}=\begin{bmatrix}
\boldsymbol{B}_{\mathrm{lin}}|\boldsymbol{A}_{\mathrm{lin}}\,\boldsymbol{B}_{\mathrm{lin}}|\boldsymbol{A}_{\mathrm{lin}}^{2}\,\boldsymbol{B}_{\mathrm{lin}}|\dots|\boldsymbol{A}_{\mathrm{lin}}^{n-1}\boldsymbol{B}_{\mathrm{lin}}
\end{bmatrix},
\label{contr_matrix}
\end{align}
\end{linenomath}
the linear control system $\dot{\boldsymbol{z}}=\boldsymbol{A}_{\mathrm{lin}}z+\boldsymbol{B}_{\mathrm{lin}}w$ is
controllable if and only if the rank of the matrix $\mathsf{C}_{\mathrm{Kal}}$ is equal to $n$.
As a consequence of this we have the following classical theorem concerning the STLC of nonlinear control systems \cite{Coron2007}.
\begin{theorem}\label{thm:lin_STLC}
    Let $(z_{\mathrm{eq}},w_{\mathrm{eq}})$ be an equilibrium point of the control system \eqref{eq: control system}. Let us assume that the linearized control system of the control system \eqref{eq: control system} at $(z_{\mathrm{eq}},w_{\mathrm{eq}})$ is controllable. Then the nonlinear control system is
small-time locally controllable at $(z_{\mathrm{eq}},w_{\mathrm{eq}})$.
\end{theorem}
\section{Controllability results for the remodelling problem}
\label{sec_under}
The state-space representation derived in Section \ref{sec_state} provides a five-dimensional non-linear control system whose drift vector field is obtained from the explicit inversion of the state matrix. Although the resulting equations are completely determined, the analytical expressions of the corresponding vector fields and of their iterated Lie brackets are algebraically cumbersome and offer limited physical insight. For this reason, rather than pursuing a fully general controllability analysis of the three-dimensional model, we investigate controllability properties on suitably selected invariant classes of infinitesimal strain fields. These restrictions preserve the essential non-linear coupling between remodelling stretches and principal directions while leading to reduced systems that remain analytically tractable. 

Throughout this section, the remodelling state is described by five independent state variables, whereas only four control inputs are introduced. Consequently, all the controlled remodelling systems considered below are under-actuated in the sense of Geometric Control Theory. The controllability results obtained in the following sections are therefore achieved despite the number of controls being strictly smaller than the dimension of the state space.

\paragraph{Controllability of a purely shear deformation.} As a first benchmark problem, we consider a deformation state generated by a single shear component of the infinitesimal strain tensor. More precisely, we prescribe 
\begin{linenomath}
\begin{align}
\boldsymbol{\varepsilon} = \mathrm{dev}(\boldsymbol{\varepsilon})= \begin{pmatrix} 0 & 0 & \varepsilon_{13}\\ 0 & 0 & 0\\ \varepsilon_{13} & 0 & 0 \end{pmatrix},
\end{align}
\end{linenomath}
where $\varepsilon_{13}\neq 0$ denotes the infinitesimal shear strain in the $(X_1,X_3)$-plane and, therefore, the considered deformation corresponds to a pure shear state acting in the $(X_1,X_3)$-plane, with no normal deformation along the coordinate axes. By applying Theorem \ref{thm_equilibrium}, the equilibrium remodelling configuration is 
\begin{linenomath}
\begin{align}\label{equilibrium13}
\bar{z}_{\rm eq} = \left(p_{\rm 1,eq},p_{\rm 2,eq},\bar\alpha_{\rm eq},\bar\beta_{\rm eq},\bar\gamma_{\rm eq}\right)^{\rm T}= \left( \varepsilon_{13},0, \pi, \pi/4,\pi \right)^{\rm T},
\end{align}
\end{linenomath}
and, since $\varepsilon_{13}\neq0$, the three principal remodelling stretches $\Lambda_1=\varepsilon_{13}, \Lambda_2=0$ and $\Lambda_3=-\varepsilon_{13}$ are pairwise distinct. Consequently, the equilibrium is regular and belongs to the open subset in which the state matrix $\boldsymbol M$ is invertible. To investigate controllability, we introduce four independent control inputs acting directly on the state variables $p_1$, $p_2$, $\beta$ and $\gamma$. The resulting affine control system reads 
\begin{linenomath}
\begin{align}
\boldsymbol M(z)\dot{\boldsymbol z}
=
\boldsymbol b(z;\varepsilon_{\rm dev})+w_{1}\boldsymbol{e}_{1}+w_{2}\boldsymbol{e}_{2}+w_{4}\boldsymbol{e}_{4}+w_{5}\boldsymbol{e}_{5},
\label{eq_state_space_remodelling_ben}
\end{align}
\end{linenomath}
where $\{\boldsymbol{e}_{i}\}_{i=1,...,5}$ is the canonical basis of $\mathbb{R}^{5}$. In its explicit form, the controlled system in Equation \eqref{eq_state_space_remodelling_ben} reads
\begin{linenomath}
\begin{align}
\dot{\boldsymbol{z}} = \boldsymbol{f}_0(z;\varepsilon_{13}) +w_{1}\boldsymbol{f}_{1}(z)+w_{2}\boldsymbol{f}_{2}(z)+w_{4}\boldsymbol{f}_{4}(z)+w_{5}\boldsymbol{f}_{5}(z),
\label{eq_state_space_remodelling_ben2}
\end{align}
\end{linenomath}
where $\boldsymbol{f}_0$ is the drift term in Equation \eqref{eq_explicit} and $\boldsymbol{f}_{j}(z)=\boldsymbol M^{-1}(z)\boldsymbol{e}_{j}$, $j=1,2,4,5$. Since the state space has dimension five and only four control inputs are employed, the system is under-actuated. We therefore analyse whether the non-linear coupling of the remodelling law is sufficient to recover controllability despite the lack of direct actuation of all state variables. To establish small-time local controllability, we invoke Theorem~\ref{thm:lin_STLC}. The Jacobian matrices of the linearized control system at the equilibrium point $(z_{\rm eq},w_{\rm eq})$ are given by 
\begin{linenomath}
\begin{align} 
\boldsymbol{A}_{\mathrm{lin}} = \nabla_z\boldsymbol{f}_0(z_{\rm eq},w_{\rm eq}) = \begin{pmatrix} -\dfrac{\mu}{b}&0 & 0 & 0 &0\\ 0&0&-\dfrac{\varepsilon_{13}\mu}{\sqrt{2}b}&0&-\dfrac{\varepsilon_{13}\mu}{b}\\ 0&0&-\dfrac{\mu}{b}&0&0\\ 0&-\dfrac{\mu}{2b\varepsilon_{13}}&0&0&0\\ 0&0&\dfrac{\mu}{\sqrt{2}b}&-\dfrac{2\mu}{b}&0 \end{pmatrix}, \end{align} 
\end{linenomath} 
and 
\begin{linenomath} 
\begin{align} 
\boldsymbol{B}_{\mathrm{lin}} = \nabla_w\boldsymbol{\mathcal F}(z_{\rm eq},w_{\rm eq}) = \begin{pmatrix} \boldsymbol{f}_{1}(z_{\rm eq}) & \boldsymbol{f}_{2}(z_{\rm eq}) & \boldsymbol{f}_{4}(z_{\rm eq}) & \boldsymbol{f}_{5}(z_{\rm eq}) \end{pmatrix}, \end{align} 
\end{linenomath} 
where 
\begin{linenomath} 
\begin{align} 
\boldsymbol{f}_{1}(z_{\rm eq})=\boldsymbol e_{1}, \quad \boldsymbol{f}_{2}(z_{\rm eq})=\boldsymbol e_{2}, \quad \boldsymbol{f}_{4}(z_{\rm eq})=\dfrac{1}{2\varepsilon_{13}}\boldsymbol e_{4}, \quad \boldsymbol{f}_{5}(z_{\rm eq})=\dfrac{\sqrt{2}}{\varepsilon_{13}}\boldsymbol e_{3} -\dfrac{1}{\varepsilon_{13}}\boldsymbol e_{5}. 
\end{align} 
\end{linenomath} 
To verify the Kalman criterion, it is sufficient to show that the Kalman matrix $\mathsf C_{\mathrm{Kal}}$ has full rank. To this end, we consider the following $5\times5$ sub-matrix of $\mathsf C_{\mathrm{Kal}}$
\begin{linenomath} 
\begin{align} 
\begin{pmatrix} 
\boldsymbol{B}_{\mathrm{lin}} \mid \boldsymbol{A}_{\mathrm{lin}}\boldsymbol{f}_{4}(z_{\rm eq}) \end{pmatrix}, 
\end{align} 
\end{linenomath} 
whose determinant is 
\begin{linenomath} 
\begin{align} 
\det \begin{pmatrix} \boldsymbol{B}_{\mathrm{lin}} \mid \boldsymbol{A}_{\mathrm{lin}}\boldsymbol{f}_{4}(z_{\rm eq}) \end{pmatrix} = \dfrac{\mu}{\sqrt{2}\,b\,(\varepsilon_{13})^{3}}. 
\end{align} 
\end{linenomath} 
Since $\varepsilon_{13}\neq0$ by assumption, the above determinant never vanishes. Therefore, the Kalman matrix has full rank, 
\begin{linenomath} 
\begin{align} 
\mathrm{rank}\,\mathsf C_{\mathrm{Kal}}=5,
\end{align} 
\end{linenomath} 
and the linearized control system is controllable. By Theorem~\ref{thm:lin_STLC}, it follows that the nonlinear remodelling system is small-time locally controllable at the equilibrium point under consideration.
{\paragraph{Numerical simulations.} The purpose of the following simulation is to illustrate, in a fully nonlinear setting, that the remodelling system can be steered in directions that are not directly spanned by any of the control vector fields $\boldsymbol{f}_1,\dots,\boldsymbol{f}_4$ themselves, but only by their Lie brackets. In particular, we show that it is possible to move the state along the $\gamma$ direction, even though none of the admissible controls act directly on $\gamma$. To this end, we numerically integrate the fully nonlinear remodelling equations over the interval $t\in[0,\bar{t}]$, with $\bar{t}=0.01 \mathrm{\,s}$, and set $\varepsilon_{13}=0.5$ and $\mu/b=1/30$. Starting from the equilibrium state $z_{\rm eq}$ in \eqref{equilibrium13} and applying the control
\begin{equation}
\label{eq_bang_contr_w4}
w_4(t)=\begin{cases}
-1 & t\in[0,\bar{t}/2],\\[2pt]
\ \ 1 & t\in[\bar{t}/2,\bar{t}],
\end{cases}
\end{equation}
i.e.\ a bang-bang switch of amplitude two acting only on the fourth control direction. Figure~\ref{fig:remodelling_w4} reports the time evolution of the three remodelling angles $
\alpha,\beta $ and $\gamma$ over the whole interval $[0,\bar{t}]$. Instead, we do not report the trend of the remodelling stretches $p_{1}$ and $p_{2}$ since they remain almost constant and equal to the equilibrium. As predicted by the theoretical estimate
\begin{linenomath} 
\begin{align} 
z(\bar{t})-z_{\rm eq}=\frac{\bar{t}^2}{4}\,[f_0,f_4]+o(\bar{t}^2)
=\begin{pmatrix}0 &0&0&0& \dfrac{\bar{t}^2}{8}\,\varepsilon_{13}\end{pmatrix}^{\rm T}+o(\bar{t}^2),
\label{eq_est}
\end{align} 
\end{linenomath} 
the variables $p_1$, $p_2$, $\alpha$ and $\beta$ undergo only negligible oscillations about their equilibrium values and return to $z_{ \rm eq}$ at $t=\bar{t}$, while $\gamma$ is the only component exhibiting a genuine drift, increasing monotonically away from its equilibrium value over the second half of the interval, in qualitative agreement with the second-order Lie bracket estimate in Equation \eqref{eq_est}. Quantitatively, the agreement between the theoretical prediction and the simulated final state is very good: the relative error between the theoretically predicted value of $\gamma(\bar{t})$ and the one obtained from the numerical integration of the fully nonlinear system is only $3.3\%$, confirming the validity of the asymptotic expansion for small $\bar{t}$ and, more generally, the capability of the control $w_4$ to selectively drive the remodelling state along the direction associated with $\varepsilon_{13}$.

\begin{figure}[htbp]
    \centering
        \includegraphics[scale=0.4]{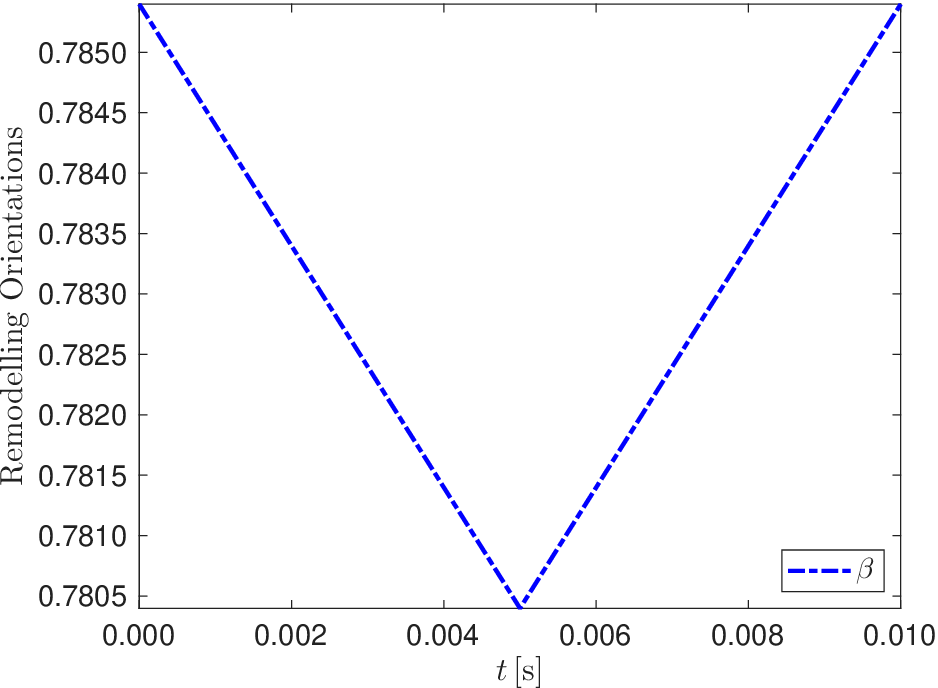}
        \includegraphics[scale=0.4]{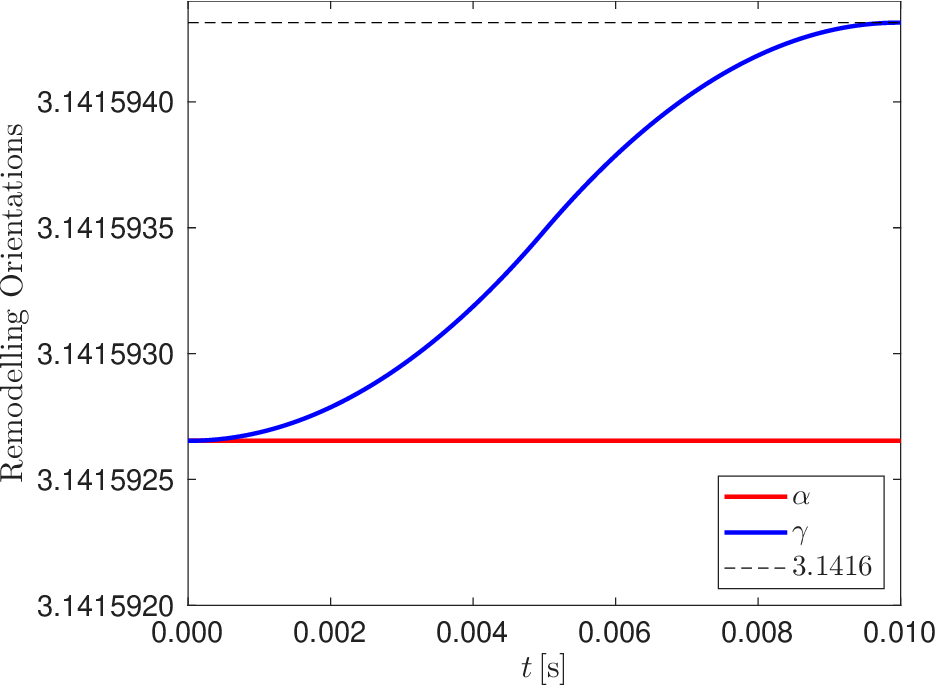}
    \caption{Time evolution of the remodelling variables $(\alpha,\beta,\gamma)$ over $t\in[0,\bar{t}]$, obtained by numerical integration of the fully nonlinear remodelling equations under the control $w_4$, starting from the equilibrium point. The other remodelling variables $p_1$ and $p_2$ exhibit only negligible variations and are therefore omitted. }
    \label{fig:remodelling_w4}
\end{figure}}
\noindent As a further test, we now combine two controls, namely
\begin{linenomath} 
\begin{align} 
w_4(t)=\begin{cases}
-\dfrac{b\varepsilon_{13}}{\sqrt{2}\mu} & t\in[0,\bar{t}/2],\\[6pt]
\ \ \dfrac{b\varepsilon_{13}}{\sqrt{2}\mu} & t\in[\bar{t}/2,\bar{t}],
\end{cases}
\qquad\text{and}\qquad
w_5(t)=\dfrac{\varepsilon_{13}\,\bar{t}}{4\sqrt{2}},
\label{eq_est}
\end{align} 
\end{linenomath} 
and integrate again the fully nonlinear remodelling equations over $t\in[0,\bar{t}]$, starting from the equilibrium state $z_{\rm eq}$, using the same parameters as in the previous simulation. From a theoretical point of view, the combined action of $w_4$ and $w_5$ produces a change in the remodelling state, to leading order, given by
\begin{linenomath} 
\begin{align} 
z(\bar{t})-z_{\rm eq}=\frac{b\varepsilon_{13}}{\sqrt{2}\mu}\,\frac{\bar{t}^2}{4}\,[\boldsymbol{f}_0,\boldsymbol{f}_4]\big|_{z_{\rm eq}}
+\frac{\varepsilon_{13}\bar{t}}{4\sqrt{2}}\,\bar{t}\,\boldsymbol{f}_5+o(\bar{t}^2)
=\begin{pmatrix}0 & 0 & \dfrac{\bar{t}^2}{4} & 0 &0\end{pmatrix}^{\rm T}+o(\bar{t}^2),
\label{eq_est}
\end{align} 
\end{linenomath} 
so that, unlike the previous case, the only variable predicted to undergo a non-negligible change is now $\alpha$, while $p_1$, $p_2$, $\beta$ and $\gamma$ are expected to remain, to this order, at their equilibrium values. Figure~\ref{fig:alpha_w4w5} shows the numerically computed evolution of $\alpha(t)$ under the combined controls $w_4$ and $w_5$, confirming this prediction: $\alpha$ drifts away from its equilibrium value $\pi$ over the interval $[0,\bar{t}]$, while the remaining variables (not shown, as their variation is negligible) stay essentially unchanged. The relative error between the theoretically predicted value of $\alpha(\bar{t})$ and the value obtained from the numerical integration of the fully nonlinear system is $3.4\%$, once again showing very good agreement between the theoretical prediction and the numerical simulations, and confirming that the pair of controls $(w_4,w_5)$ can be used to selectively steer the remodelling state along the $\alpha$ direction.
\begin{figure}[htbp]
    \centering
    \includegraphics[scale=0.4]{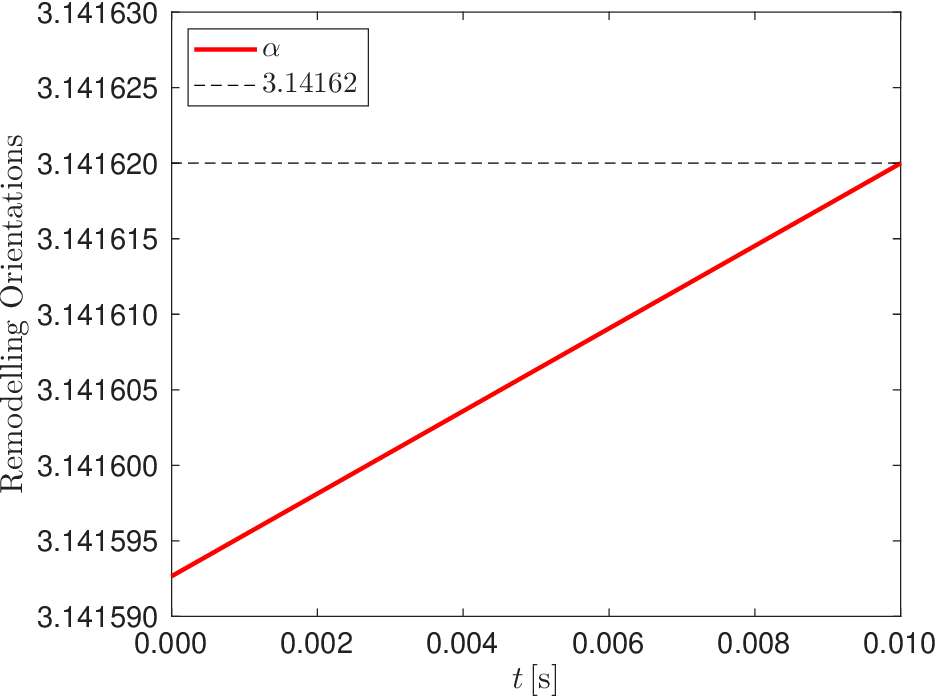}
    \caption{Time evolution of the variable $\alpha(t)$ over $t\in[0,\bar{t}]$, obtained by numerical integration of the fully nonlinear remodelling equations under the combined controls $w_4$ and $w_5$, starting from the equilibrium point. The other remodelling variables $(p_1,p_2,\beta,\gamma)$ exhibit only negligible variations and are therefore omitted.\vspace{-1cm}}
    \label{fig:alpha_w4w5}
\end{figure}

\paragraph{Controllability of a mixed shear-normal deformation.} As a second benchmark case, we consider a deformation state characterised by the simultaneous presence of a shear component in the $(X_1,X_3)$-plane and a normal deformation along the $X_2$-direction, namely
\begin{linenomath}
\begin{align}
\boldsymbol{\varepsilon} =
\begin{pmatrix}
0 & 0 & \varepsilon\\
0 & \varepsilon_{22} & 0\\
\varepsilon & 0 & 0
\end{pmatrix},
\qquad
\mathrm{dev}(\boldsymbol{\varepsilon})=
\begin{pmatrix}
-\dfrac{\varepsilon_{22}}{3} & 0 & \varepsilon\\
0 & \dfrac{2\varepsilon_{22}}{3} & 0\\
\varepsilon & 0 & -\dfrac{\varepsilon_{22}}{3}
\end{pmatrix}.
\end{align}
\end{linenomath}
Here, $\varepsilon$ denotes the infinitesimal shear strain in the $(X_1,X_3)$-plane, whereas $\varepsilon_{22}$ represents the infinitesimal normal strain associated with fibres initially aligned with the $X_2$-direction. The resulting deformation combines a transverse shear distortion with an axial extension or compression along the second coordinate axis.

According to Theorem~\ref{thm_equilibrium}, the corresponding equilibrium configuration is
\begin{linenomath}
\begin{align}
\label{equilibrium22}
\bar z_{\rm eq}
=
\left(
 p_{\rm 1,eq},
 p_{\rm 2,eq},
 \bar\alpha_{\rm eq},
 \bar\beta_{\rm eq},
 \bar\gamma_{\rm eq}
\right)^{\rm T}
=
\left(
 \varepsilon-\dfrac{\varepsilon_{22}}{3},
 \frac{2}{3}\varepsilon_{22},
 \pi,
 \frac{\pi}{4},
 \pi
\right)^{\rm T}.
\end{align}
\end{linenomath}
This configuration belongs to the regular region where the state matrix $\boldsymbol M$ is invertible and the explicit state-space representation is therefore well defined. To investigate controllability, four independent control inputs are introduced and assumed to act directly on the state variables $p_1$, $p_2$, $\beta$, and $\gamma$. The resulting controlled remodelling law is given by Equations~\eqref{eq_state_space_remodelling_ben} and~\eqref{eq_state_space_remodelling_ben2}.

As in the previous benchmark, the Jacobian matrix of the drift evaluated at equilibrium reads
\begin{linenomath}
\begin{align}
\boldsymbol A_{\mathrm{lin}}
=
\nabla_z\boldsymbol f_0(z_{\rm eq},w_{\rm eq})
=
\begin{pmatrix}
-\dfrac{\mu}{b} & 0 & 0 & 0 & 0 \\
0 & 0 & \dfrac{\mu}{\sqrt{2}b}(\varepsilon_{22}-\varepsilon) & 0 & \dfrac{\mu}{b}(\varepsilon_{22}-\varepsilon) \\
0 & 0 & -\dfrac{\mu}{b} & 0 & 0 \\
0 & -\dfrac{\mu}{2b\,\varepsilon} & 0 & 0 & 0 \\
0 & 0 & \dfrac{\mu}{\sqrt{2}b} & -\dfrac{2\varepsilon\mu}{b(\varepsilon-\varepsilon_{22})} & 0
\end{pmatrix}.
\end{align}
\end{linenomath}
The corresponding input matrix is
\begin{linenomath}
\begin{align}
\boldsymbol B_{\mathrm{lin}}
=
\nabla_w\boldsymbol{\mathcal F}(z_{\rm eq},w_{\rm eq})
=
\begin{pmatrix}
\boldsymbol f_1(z_{\rm eq}) &
\boldsymbol f_2(z_{\rm eq}) &
\boldsymbol f_4(z_{\rm eq}) &
\boldsymbol f_5(z_{\rm eq})
\end{pmatrix},
\end{align}
\end{linenomath}
where
\begin{linenomath}
\begin{align}
\boldsymbol f_1(z_{\rm eq}) = \boldsymbol e_1, \quad \boldsymbol f_2(z_{\rm eq}) = \boldsymbol e_2, \quad \boldsymbol f_4(z_{\rm eq}) = \dfrac{1}{2\varepsilon}\,\boldsymbol e_4, \quad \boldsymbol f_5(z_{\rm eq}) = \dfrac{\sqrt{2}}{\varepsilon+\varepsilon_{22}}\,\boldsymbol e_3
-\dfrac{1}{\varepsilon+\varepsilon_{22}}\,\boldsymbol e_5.
\end{align}
\end{linenomath}
To apply Theorem~\ref{thm:lin_STLC}, it is sufficient to verify that the Kalman controllability matrix has full rank. To this end, we consider the following $5\times5$ sub-matrix of $\mathsf C_{\mathrm{Kal}}$:
\begin{linenomath}
\begin{align}
\begin{pmatrix}
\boldsymbol B_{\mathrm{lin}}
\;\big|\;
\boldsymbol A_{\mathrm{lin}}\boldsymbol f_4(z_{\rm eq})
\end{pmatrix}.
\end{align}
\end{linenomath}
A direct calculation yields
\begin{linenomath}
\begin{align}
\det
\begin{pmatrix}
\boldsymbol B_{\mathrm{lin}}
\;\big|\;
\boldsymbol A_{\mathrm{lin}}\boldsymbol f_4(z_{\rm eq})
\end{pmatrix}
=
\dfrac{\mu}
{\sqrt{2}\,b\,\varepsilon\,(\varepsilon^{2}-\varepsilon_{22}^{2})}
\neq 0.
\end{align}
\end{linenomath}
Hence, the Kalman controllability matrix has maximal rank equal to the dimension of the state space, namely $n=5$. It follows that the linearised control system is controllable. Consequently, by Theorem~\ref{thm:lin_STLC}, the nonlinear remodelling system is small-time locally controllable at the equilibrium configuration $\bar z_{\rm eq}$ 

\paragraph{Numerical simulation for mixed shear-normal deformation.}
We consider now the case of a mixed shear-normal deformation, with parameters $\bar{t}=0.01$, $\mu/b=1/30$, $\varepsilon=1/2$ and $\varepsilon_{22}=1/4$, and integrate again the fully nonlinear remodelling equations starting from the equilibrium state $z_{\rm eq}$ in \eqref{equilibrium22}, using the same bang-bang control of formula \eqref{eq_bang_contr_w4} acting only on the fourth control direction. In this case the theoretical prediction, again obtained via the Lie bracket $[f_0,f_4]$, gives
\begin{linenomath}
\begin{align}
z(\tau)=z_{\rm eq}+\frac{\bar{t}^2}{4}\,[\boldsymbol{f}_0,\boldsymbol{f}_4]
=\begin{pmatrix}0& 0 &0 &0& \dfrac{\mu\,\bar{t}^2}{4b}\,(\varepsilon-\varepsilon_{22})\end{pmatrix}^{\rm T}+o(\bar{t}^2),
\end{align}
\end{linenomath}
so that, as before, only the $\gamma$ component is expected to undergo a non-negligible change. Figure~\ref{fig:gamma_caseC} shows the numerically computed evolution of $\gamma(t)$, confirming this prediction. As in the previous cases, the agreement between theory and numerics is very good: the relative error between the theoretically predicted value of $\gamma(\bar{t})$ and the one obtained from the numerical integration of the fully nonlinear system is $3.3\%$.

\begin{figure}[htbp]
    \centering
    \includegraphics[scale=0.4]{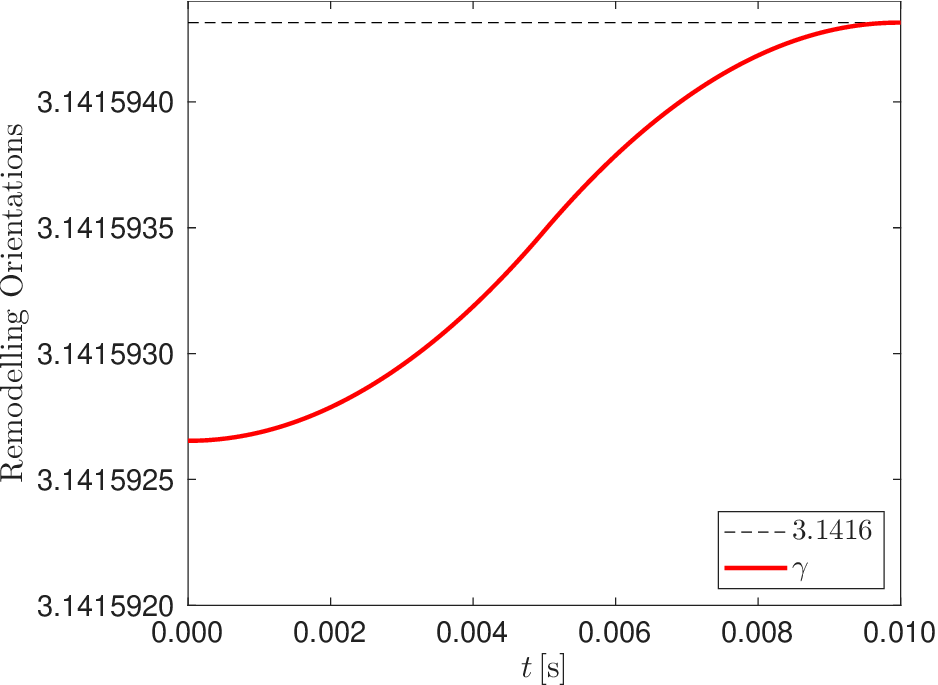}
    \caption{Time evolution of the variable $\gamma(t)$ over $t\in[0,\bar{t}]$, $\bar{t}=0.01$, for the mixed shear-normal deformation case ($\mu/b=1/30$, $\varepsilon=1/2$, $\varepsilon_{22}=1/4$), obtained by numerical integration of the fully nonlinear remodelling equations under the control $w_4$, starting from the equilibrium point. 
    The other remodelling variables exhibit only negligible variations and are therefore omitted.}
    \label{fig:gamma_caseC}
\end{figure}\vspace{-1cm}
\section{Planar reduction and non-controllability results}
\label{sec:planar_state_space_equation}

We now consider an invariant planar restriction of the linearised
deformation and remodelling problem developed in the previous sections
and apply the tools of Geometric Control Theory recalled in
Section~\ref{sec_control}. This choice is consistent with the structure
of the manuscript and allows the reduced problem to serve as an
analytical benchmark for the three-dimensional formulation. We do not
repeat all the formal steps underlying the general theory, since the
present construction follows the same procedure and introduces no
additional conceptual elements. Instead, we collect the main ingredients
required for the subsequent controllability analysis.

As in the three-dimensional setting, remodelling is assumed to be
isochoric. After eliminating the polar remodelling rotation by setting
$\boldsymbol R_{\rm p}=\boldsymbol I$, we restrict the problem to the
$X_1$--$X_2$ plane. Accordingly, the reference configuration
$\mathscr B$ is contained in this plane and its material points are
represented as $X=(X_1,X_2)$. The remodelling state is described by the
principal stretch $p$ and by the angle $\theta$, which determines the
orientation of the associated principal directions within the plane. We
therefore set
\begin{linenomath}
\begin{align}
\boldsymbol\Lambda_{\rm p}(z)
&=
\begin{pmatrix}
p&0&0\\[2pt]
0&-p&0\\[2pt]
0&0&0
\end{pmatrix},
&
\boldsymbol Q_{\rm p}(z)
&=
\begin{pmatrix}
\cos\theta&-\sin\theta&0\\[2pt]
\sin\theta&\cos\theta&0\\[2pt]
0&0&1
\end{pmatrix},
\label{eq:planar_Lambda_Q}
\end{align}
\end{linenomath}
and a direct calculation gives
\begin{linenomath}
\begin{align}
\dot{\boldsymbol\Lambda}_{\rm p}
+
\boldsymbol\Omega\boldsymbol\Lambda_{\rm p}
-
\boldsymbol\Lambda_{\rm p}\boldsymbol\Omega
=
\begin{pmatrix}
\dot p&2p\dot\theta&0\\[3pt]
2p\dot\theta&-\dot p&0\\[3pt]
0&0&0
\end{pmatrix}.
\label{eq:planar_evolution_lhs}
\end{align}
\end{linenomath}
For the evolution to remain confined to the selected planar subspace, the driving strain must have a deviatoric part with vanishing $13$, $23$, and $33$ components. We therefore restrict the infinitesimal strain tensor accordingly and obtain
\begin{linenomath}
\begin{align}
\mathrm{dev}(\boldsymbol\varepsilon)
=
\begin{pmatrix}
\dfrac{\varepsilon_{11}-\varepsilon_{22}}{2}
&\varepsilon_{12}&0\\[3pt]
\varepsilon_{12}
&-\dfrac{\varepsilon_{11}-\varepsilon_{22}}{2}&0\\[3pt]
0&0&0
\end{pmatrix},
 \qquad \boldsymbol\varepsilon
=
\begin{pmatrix}
\varepsilon_{11}&\varepsilon_{12}&0\\[3pt]
\varepsilon_{12}&\varepsilon_{22}&0\\[3pt]
0&0&\dfrac{\varepsilon_{11}+\varepsilon_{22}}{2}
\end{pmatrix}.
\label{eq:planar_deviatoric_strain_components}
\end{align}
\end{linenomath}
Thus, the prescribed form of the three-dimensional strain has exactly the planar, symmetric, and traceless driving part required by the reduced remodelling law. The selected subspace is consequently invariant under the evolution.

\begin{remark}[Kinematic realisability of the planar restriction]
The prescribed strain field can be generated on the material plane
$X_3=0$ by a three-dimensional displacement field defined in a
neighbourhood of that plane. In particular, consider
\begin{linenomath}
\begin{align}
\boldsymbol u(X_1,X_2,X_3)
=
\begin{pmatrix}
u_1(X_1,X_2)
\\[1mm]
u_2(X_1,X_2)
\\[1mm]
\dfrac{\varepsilon_{11}(X_1,X_2)
+\varepsilon_{22}(X_1,X_2)}{2}X_3
\end{pmatrix},
\label{eq:planar_displacement_extension}
\end{align}
\end{linenomath}
where the in-plane displacement components $u_1$ and $u_2$ generate
$\varepsilon_{11}$, $\varepsilon_{22}$ and $\varepsilon_{12}$. On the
plane $X_3=0$, one obtains Equation \eqref{eq:planar_deviatoric_strain_components}. This construction is valid provided that the in-plane strain components
satisfy the corresponding compatibility condition. If
$\varepsilon_{11}+\varepsilon_{22}$ varies with $X_1$ or $X_2$,
transverse shear components may arise away from $X_3=0$. The present
formulation is therefore equivalent to an intrinsically two-dimensional
problem on the selected material plane and for the reduced remodelling
dynamics, but not necessarily for a three-dimensional body of
finite thickness. 
\end{remark}

\noindent
The reduction of the linearised remodelling law in
Equation~\eqref{eq_final_linearized_evolution} then yields the following
two linearly independent scalar equations
\begin{linenomath}
\begin{subequations}
\begin{align}
\dot p
&=
\frac{\mu}{b}
\left[
\frac{\varepsilon_{11}-\varepsilon_{22}}{2}\cos(2\theta)
+\varepsilon_{12}\sin(2\theta)
-p
\right],
\label{eq:planar_p_implicit}
\\[6pt]
2p\dot\theta
&=
\frac{\mu}{b}
\left[
-\frac{\varepsilon_{11}-\varepsilon_{22}}{2}\sin(2\theta)
+\varepsilon_{12}\cos(2\theta)
\right].
\label{eq:planar_theta_implicit}
\end{align}
\end{subequations}
\end{linenomath}
The evolution equations are supplemented by the initial conditions
\begin{linenomath}
\begin{align}
p(X,t_{\rm in})
=
p^{\rm (in)}(X),
\qquad
\theta(X,t_{\rm in})
=
\theta^{\rm (in)}(X),
\qquad
X\in\mathscr B.
\label{eq:planar_remodelling_initial_conditions}
\end{align}
\end{linenomath}
If $p^{\rm (in)}(X)=0$ at a material point, the two principal remodelling
stretches coincide and no preferred principal direction exists at that
point. Consequently, $\theta^{\rm (in)}(X)$ is not physically identifiable, and
the spectral evolution equations must be interpreted locally on the
subset where $p\neq0$. Before proceeding, we remark that the deformation boundary-value problem
defined by Equations~\eqref{eq_balance_linear_mom},
\eqref{BVP_tractions} and\eqref{BVP_Dirichlet} must be reduced
accordingly to the present planar setting.

\paragraph{State-space representation.}
We introduce the state array and its time derivative as
\begin{linenomath}
\begin{align}
z
&:=
\begin{bmatrix}
p &
\theta
\end{bmatrix}^{\rm T},
&
\dot{\boldsymbol z}
&:=
\begin{bmatrix}
\dot p &
\dot\theta
\end{bmatrix}^{\rm T}.
\label{eq:planar_state_array}
\end{align}
\end{linenomath}
Equations~\eqref{eq:planar_p_implicit} and
\eqref{eq:planar_theta_implicit} can be written as
\begin{linenomath}
\begin{align}
\boldsymbol A(z)
\dot{\boldsymbol z}
=
\boldsymbol b(z;\varepsilon_{\rm dev}),
\label{eq:planar_state_space_equation}
\end{align}
\end{linenomath}
where
\begin{linenomath}
\begin{align}
\boldsymbol A(z)
=
\begin{pmatrix}
1&0\\[3pt]
0&2p
\end{pmatrix} \quad \mbox{ and } \quad 
\boldsymbol b(z;\varepsilon_{\rm dev})
=
\frac{\mu}{b}
\begin{pmatrix}
\dfrac{\varepsilon_{11}-\varepsilon_{22}}{2}\cos(2\theta)
+\varepsilon_{12}\sin(2\theta)-p
\\[4mm]
-\dfrac{\varepsilon_{11}-\varepsilon_{22}}{2}\sin(2\theta)
+\varepsilon_{12}\cos(2\theta)
\end{pmatrix}.
\label{eq:planar_state_vector}
\end{align}
\end{linenomath}
This convention retains the factor $\mu/b$ on the right-hand side, in
agreement with the three-dimensional state-space representation. The
determinant of $\boldsymbol A$ is
\begin{linenomath}
\begin{align}
\det\boldsymbol A
=
2p.
\label{eq:planar_state_matrix_determinant}
\end{align}
\end{linenomath}
Therefore, the spectral state-space representation is locally invertible
if and only if $p\neq0$. On this regular subset, $\boldsymbol A^{-1}(z)=\mathrm{diag}\{1,1/2p\}$ and the evolution law is written as $\dot{\boldsymbol{z}}=\boldsymbol{f}_{0}(z,\varepsilon_{\rm dev})$, where $\boldsymbol{f}_{0}$ is a two-dimensional drift term having expression
\begin{linenomath}
\begin{align}
\boldsymbol{f}_{0}(z;\varepsilon_{\rm dev})=\boldsymbol M^{-1}(z)\boldsymbol{b}(z,\varepsilon_{\rm dev})=\dfrac{\mu}{b}\begin{pmatrix}
\dfrac{\varepsilon_{11}-\varepsilon_{22}}{2}\cos(2\theta)
+\varepsilon_{12}\sin(2\theta)-p
\\
\\
-\dfrac{\varepsilon_{11}-\varepsilon_{22}}{4p}\sin(2\theta)
+\dfrac{\varepsilon_{12}}{2p}\cos(2\theta)
\end{pmatrix}.
\label{eq_drift_2D}
\end{align}
\end{linenomath}
The loss of invertibility at $p=0$ reflects the degeneracy of the
spectral coordinates, with the two principal remodelling
stretches coinciding and no preferred principal direction can be
identified. Thus, the singularity of $\boldsymbol A^{-1}$ does not indicate a singularity of the underlying tensorial
remodelling process. Finally, the spectral coordinates are not globally unique. In particular,
\begin{linenomath}
\begin{align}
(p,\theta)
&\sim
(p,\theta+\pi),
&
(p,\theta)
&\sim
\left(-p,\theta+\frac{\pi}{2}\right),
\label{eq:planar_spectral_equivalence}
\end{align}
\end{linenomath}
since these pairs reconstruct the same symmetric deviatoric remodelling
tensor. A unique representation therefore requires an additional
convention, such as $p\geq0$ together with $\theta\in[0,\pi)$.

\paragraph{Equilibrium points.} Assuming the strain components to be constant in time and by applying Theorem \eqref{thm_equilibrium}, we obtain that the equilibrium configuration of the considered problem is unique and given by
\begin{linenomath}
\begin{align}
\bar{p}_{\rm eq}=\sqrt{\left(\dfrac{\varepsilon_{11}-\varepsilon_{22}}{2}\right)^{2}+(\varepsilon_{12})^2} \quad \mbox{ and } \quad \bar{\theta}_{\rm eq}=\dfrac{1}{2}\arctan\left(\dfrac{2\varepsilon_{12}}{\varepsilon_{11}-\varepsilon_{22}}\right) \, \mathrm{mod}\,\pi.
\label{eq_equilibria_2D}
\end{align}
\end{linenomath}
If $p_{\rm eq}^{(0)}>0$, this equilibrium is unique modulo the spectral equivalences in Equation~\eqref{eq:planar_spectral_equivalence}. If $p_{\rm eq}^{(0)}=0$, the equilibrium remodelling tensor is isotropic and the angular coordinate is not identifiable. In the following, the controllability analysis is therefore restricted to the regular case $p_{\rm eq}^{(0)}>0$. Moreover, the equilibrium reported in Equation \eqref{eq_equilibria_2D} is a zero of the drift term in Equation \eqref{eq_drift_2D} and, hence, is an equilibrium of the non-controlled remodelling law. 

\paragraph{Non-Controllability through $p$.} We now consider the controllability of the remodelling law in Equations \eqref{eq:planar_state_space_equation} and \eqref{eq:planar_state_vector} by introducing a control function $w_{1}$ which directly drives the evolution of $p$. Hence, we write Equation \eqref{eq:planar_state_space_equation} in the controlled form 
\begin{linenomath}
\begin{align}
\boldsymbol M(z)
\dot{\boldsymbol z}
=
\boldsymbol b(z;\varepsilon_{\rm dev})+w_{1}\dfrac{\mu}{b}\begin{bmatrix}
1 \\ 0
\end{bmatrix},
\label{eq:planar_state_space_equationb}
\end{align}
\end{linenomath}
and, accordingly, its explicit version (compare with Equation \eqref{eq_drift_2D}) is
\begin{linenomath}
\begin{align}
\dot{\boldsymbol{z}}=\boldsymbol{f}_{0}(z;\varepsilon_{\rm dev})+w_{1}\boldsymbol{f}_{1}(z), \qquad \boldsymbol{f}_{1}(z)=\dfrac{\mu}{b}\boldsymbol{A}^{-1}(z)\begin{bmatrix}
1 \\0
\end{bmatrix}=\dfrac{\mu}{b}\begin{bmatrix}
1 \\0
\end{bmatrix}.
\end{align}
\end{linenomath}
We then compute the Lie bracket
\begin{linenomath}
\begin{align}
[\boldsymbol{f}_{0},\boldsymbol{f}_{1}](\bar{z}_{\rm eq})=\dfrac{\mu^2}{b^2}\begin{bmatrix}
1\\ 0
\end{bmatrix}.
\end{align}
\end{linenomath}
which is clearly a multiple of $\boldsymbol{f}_1$. Therefore the Lie algebra generated by $\boldsymbol f_0$ and $\boldsymbol f_1$ is one-dimensional and cannot span $\mathbb R^2$. Consequently, the Lie Algebra Rank Condition is violated at the equilibrium point. By Proposition~\ref{propositionLARC}, the system cannot be small-time locally controllable.

\paragraph{Non-Controllability through $\theta$.} We now consider a control function $w_{2}$ which directly steers the remodelling law through the evolution of $\theta$. This way, we write Equation \eqref{eq:planar_state_space_equation} as
\begin{linenomath}
\begin{align}
\boldsymbol M(z)
\dot{\boldsymbol z}
=
\boldsymbol b(z;\varepsilon_{\rm dev})+w_{2}\dfrac{\mu}{b}\begin{bmatrix}
0 \\ 1
\end{bmatrix},
\label{eq:planar_state_space_equationc}
\end{align}
\end{linenomath}
and, accordingly, its explicit version (compare with Equation \eqref{eq_drift_2D}) is
\begin{linenomath}
\begin{align}
\dot{\boldsymbol{z}}=\boldsymbol{f}_{0}(z;\varepsilon_{\rm dev})+w_{2}\boldsymbol{f}_{2}(z), \qquad \boldsymbol{f}_{2}(z)=\dfrac{\mu}{b}\boldsymbol M^{-1}(z)\begin{bmatrix}
0 \\1
\end{bmatrix}=\dfrac{\mu}{2bp}\begin{bmatrix}
0 \\1
\end{bmatrix}.
\end{align}
\end{linenomath}
By computing the Lie bracket
\begin{linenomath}
\begin{align}
[\boldsymbol{f}_{0},\boldsymbol{f}_{2}](\bar{z}_{\rm eq})=\dfrac{\mu^2}{b^2}\begin{bmatrix}
0 \\ 1
\end{bmatrix},
\end{align}
\end{linenomath}
which is a scalar multiple of $\boldsymbol f_2(\bar z_{\rm eq})$. Hence, the Lie algebra generated by $\boldsymbol f_0$ and $\boldsymbol f_2$ remains one-dimensional and the LARC condition cannot be satisfied. By Proposition~\ref{propositionLARC}, the system is not small-time locally controllable at the considered equilibrium.

\medskip
\noindent
The following theorem summarizes the main results.
\begin{theorem}[Planar non-controllability through stretch or angular actuation] For every infinitesimal strain tensor belonging to the planar class defined by Equation~\eqref{eq:planar_deviatoric_strain_components}, neither the control system \eqref{eq:planar_state_space_equationb}, obtained by actuation of the stretch variable $p$, nor the control system \eqref{eq:planar_state_space_equationc}, obtained by actuation of the angular variable $\theta$, is small-time locally controllable at the equilibrium configuration \eqref{eq_equilibria_2D}. \end{theorem}
\begin{proof}
    In both cases, the controlled vector field $\boldsymbol f_i$, $i=1,2$, is an eigenvector of the Jacobian $\nabla_{z} \boldsymbol f_0(\bar z_{\rm eq})$. Consequently, the Lie bracket $[\boldsymbol{f}_{0},\boldsymbol{f}_{i}](\bar{z}_{\rm eq})$ remains proportional to $\boldsymbol f_i$, and the same property holds for all higher-order iterated brackets. Therefore, the Lie algebra generated by $\boldsymbol f_0$ and $\boldsymbol f_i$ cannot span the whole $\mathbb R^2$ and remains one-dimensional. The Lie Algebra Rank Condition is thus violated at the equilibrium point and, by Proposition~\ref{propositionLARC}, the corresponding control system cannot be small-time locally controllable.
\end{proof}
\noindent
The obtained non-controllability result admits a mechanical interpretation. A perturbation acting solely on the remodelling stretch modifies the intensity of the remodelling process without generating an independent mechanism capable of controlling the orientation of the principal remodelling directions. Conversely, a perturbation acting only on the orientation affects the remodelling frame but does not provide sufficient authority to independently regulate the associated remodelling stretches. Although the two variables influence each other through the drift dynamics, this interaction remains constrained and does not produce the additional degrees of freedom required for local controllability. The planar remodelling problem therefore exhibits an intrinsic limitation of single-input actuation. In contrast with the considered three-dimensional under-actuated systems, the nonlinear coupling between stretch and orientation is not sufficient to recover the missing control direction, and controllability is lost.

\section{Conclusions}
\label{sec_concl}
In this work, we formulated a control-theoretic framework for continuum remodelling and investigated its controllability properties within the perspective of \emph{Continuum Control Theory}. By combining the Bilby--Kröner--Lee decomposition with a spectral parametrisation of the remodelling tensor, the remodelling dynamics was reformulated as a finite-dimensional nonlinear control system parametrised by the material position. This allowed the methods of Geometric Control Theory to be applied directly to the evolution of the internal material structure. The analysis was restricted to isochoric remodelling in the infinitesimal-strain regime while retaining finite rotations of the principal remodelling directions. Under these assumptions, the remodelling tensor was represented through its principal stretches and orientations, leading to a nonlinear state-space description in which the coupling between these quantities remains fully preserved. 

For selected three-dimensional strain classes, the resulting control systems were shown to be small-time locally controllable at the equilibrium point, despite being under-actuated. In these cases, the nonlinear coupling among remodelling stretches and orientations generates sufficient accessible directions in the state space to compensate for the absence of direct actuation of all state variables. Consequently, the complete remodelling state can be steered through a number of controls smaller than the dimension of the state space. 

A markedly different behaviour was found in the planar reduction. Although the reduced system retains a nonlinear coupling between remodelling stretch and orientation, controllability is lost when only one control input is available. The Lie algebra generated by the drift and the controlled direction remains one-dimensional and the Lie Algebra Rank Condition is not satisfied. As a consequence, neither stretch actuation nor angular actuation alone is sufficient to achieve small-time local controllability. This result highlights an important conceptual point. The loss of controllability is not caused by the absence of nonlinear coupling, since such coupling is still present in the planar model. Rather, it originates from the inability of the nonlinear dynamics to generate new independent directions in the state space. From a geometric viewpoint, nonlinear interactions are therefore not sufficient by themselves to guarantee controllability: what matters is their capacity to enrich the accessibility distribution of the system. From a mechanical perspective, the obtained results suggest that the controllability of remodelling processes depends not only on the number of available control actions but also on the geometric structure of the coupling between remodelling stretches and orientations. While the three-dimensional problem possesses sufficient internal structure to recover the missing control directions, the planar reduction does not. This observation suggests that dimensional reduction may fundamentally alter controllability properties and should therefore be treated with caution when designing control strategies for remodelling systems. 

Future developments will address the extension of the proposed framework to fully coupled deformation-remodelling dynamics, to time-dependent strain fields, and to continuum formulations involving partial differential equations. A further direction concerns the physical interpretation of the reconstructed controls and their relation to experimentally accessible external stimuli, with the aim of establishing a closer connection between constitutive modelling, continuum mechanics and control theory.

\enlargethispage{20pt}

\textbf{Acknowledgments:} {S.D.S. acknowledges the PNRR MUR project ``\textit{National Quantum Science and Technology Institute} (codice progetto PE00000023)''. S.D.S. and M.Z. are supported by the Italian National Group of Mathematical Physics (GNFM, INdAM).\vspace{-1.5cm}}


\vskip2pc



\vskip2pc

\bibliographystyle{RS}

\bibliography{bibliography}

@article {MMSZ,
    AUTHOR = {R. Marchello and  Morandotti and H. Shum and M.
              Zoppello},
     TITLE = {The {$N$}-link swimmer in three dimensions: controllability
              and optimality results},
   JOURNAL = {Acta Appl. Math.},
    VOLUME = {178},
      YEAR = {2022},
     PAGES = {Paper No. 6, 25},
      ISSN = {0167-8019,1572-9036},
       DOI = {10.1007/s10440-022-00480-3},
}

@Book{Micunovic2009a,
  author    = {M. Micunovic},
  publisher = {Springer New York},
  title     = {Thermomechanics of {V}iscoplasticity: {F}undamentals and {A}pplications},
  year      = {2009},
  isbn      = {9780387894904},
  doi       = {10.1007/978-0-387-89490-4},
  issn      = {1876-9896},
  journal   = {Advances in Mechanics and Mathematics},
}

@book{Coron2007,
author = "J.-M. Coron",
title = "Control and nonlinearity",
series="Mathematical Surveys and Monographs",
publisher = "AMS",
address = "Providence, RI",
year = "2007",
}

@book{bloch,
  title={Nonholonomic Mechanics and Control},
  author={Bloch, A.M.},
  publisher={Springer},
  year={2015},
  address={New York, NY}
}

@article{Grillo2023MEMOCSa,
  title={Comparison between different viewpoints on bulk growth mechanics},
  author={A. Grillo and S. {Di Stefano}},
  journal={Mathematics and Mechanics of Complex Systems},
  volume={11},
  number={2},
  pages={287--311},
  year={2023},
  publisher={Mathematical Sciences Publishers}
}

@article{Grillo2023MEMOCSb,
  title={An a posteriori approach to the mechanics of volumetric growth},
  author={A. Grillo and S. {Di Stefano}},
  journal={Mathematics and Mechanics of Complex Systems},
  volume={11},
  number={1},
  pages={57--86},
  year={2023},
  publisher={Mathematical Sciences Publishers}
}

@article{GrilloMMS2023a,
  title={A formulation of volumetric growth as a mechanical problem subjected to non-holonomic and rheonomic constraint},
  author={A. Grillo and S. {Di Stefano}},
  journal={Mathematics and Mechanics of Solids},
  volume={28},
  number={10},
  pages={2215--2241},
  year={2023},
  publisher={SAGE Publications Sage UK: London, England}
}

@book{LibroJurdjevic,
author = "V. Jurdjevic",
title = "Geometric Control Theory", 
publisher = "Cambridge University Press",
address = "Cambridge",
year = "1997",
}

@article{DiStefano2025a,
  title={A priori and a posteriori controllability of volumetric growth: a theoretical approach},
  author={S. {Di Stefano} and M. Zoppello},
  journal={Zeitschrift f{\"u}r angewandte Mathematik und Physik},
  volume={76},
  number={6},
  pages={1--27},
  year={2025},
  publisher={Springer}
}

@article{dicarlo2002a,
  title={Growth and balance},
  author={A. DiCarlo and S. Quiligotti},
  journal={Mechanics Research Communications},
  volume={29},
  number={6},
  pages={449--456},
  year={2002},
  publisher={Elsevier}
}

@article{CardinGianniottiSpiro2021,
  title={Control problems with differential constraints of higher order},
  author={Cardin, Franco and Giannotti, Cristina and Spiro, Andrea},
  journal={Nonlinear Analysis},
  volume={207},
  pages={112263},
  year={2021},
  publisher={Elsevier}
}

@article{DiStefano2026a,
  title={Optimal controllability in space of continuum adhesive systems: an a posteriori approach},
  author={Di Stefano, Salvatore and Giammarini, Alessandro and Zoppello, Marta},
  journal={Zeitschrift f{\"u}r angewandte Mathematik und Physik},
  volume={77},
  number={7},
  pages={186},
  year={2026},
  publisher={Springer}
}

@article{DiStefano2022c,
  title={An elasto-plastic biphasic model of the compression of multicellular aggregates: the influence of fluid on stress and deformation: S. Di Stefano et al.},
  author={Di Stefano, Salvatore and Giammarini, Alessandro and Giverso, Chiara and Grillo, Alfio},
  journal={Zeitschrift f{\"u}r angewandte Mathematik und Physik},
  volume={73},
  number={2},
  pages={79},
  year={2022},
  publisher={Springer}
}

@article{Cleja2000a,
  title={Eshelby's stress tensors in finite elastoplasticity},
  author={Cleja-Tigoiu, S and Maugin, GA},
  journal={Acta Mechanica},
  volume={139},
  number={1},
  pages={231--249},
  year={2000},
  publisher={Springer}
}

@article{SansonettoZoppello2020,
  author={Sansonetto, Nicola and Zoppello, Marta},
  journal={IEEE Control Systems Letters}, 
  title={On the Trajectory Generation of the Hydrodynamic Chaplygin Sleigh}, 
  year={2020},
  volume={4},
  number={4},
  pages={922-927},
  doi={10.1109/LCSYS.2020.2996763}}

@article{FPZ,
  author    = {Fass{\`o}, Francesco and Passarella, Simone and Zoppello, Marta},
  title     = {Control of locomotion systems and dynamics in relative periodic orbits},
  journal   = {Journal of Geometric Mechanics},
  volume    = {12},
  number    = {3},
  pages     = {395--420},
  year      = {2020},
  doi       = {10.3934/jgm.2020022},
  url       = {https://www.aimsciences.org/article/doi/10.3934/jgm.2020022}
}

@article{BlochObstacleAvoidance,
  title={Dynamic interpolation for obstacle avoidance on Riemannian manifolds},
  author={Bloch, Anthony and Camarinha, Margarida and Colombo, Leonardo},
  journal={International Journal of Control},
  volume={93},
  number={10},
  pages={2354--2367},
  year={2019},
  publisher={Taylor \& Francis},
  doi={10.1080/00207179.2019.1603400}
}

@article{DMPSansonetto,
  title={Dynamic Movement Primitives: Volumetric Obstacle Avoidance},
  author={Michele Ginesi and Daniele Meli and Andrea Calanca and Diego Dall’Alba and Nicola Sansonetto and Paolo Fiorini},
  journal={International Journal of Control},
  volume={93},
  number={10},
  pages={2354--2367},
  year={2019},
  publisher={Taylor \& Francis},
  doi={10.1080/00207179.2019.1603400}
}

@article{Wang2020OptimalCO,
  title={Optimal Control of a Soft CyberOctopus Arm},
  author={Tixian Wang and Udit Halder and Heng-Sheng Chang and Mattia Gazzola and Prashant G. Mehta},
  journal={2021 American Control Conference (ACC)},
  year={2020},
  pages={4757-4764},
  url={https://api.semanticscholar.org/CorpusID:222134131}
}

@book {AgrachevBook,
    AUTHOR = {Agrachev, Andrei A. and Sachkov, Yuri L.},
     TITLE = {Control theory from the geometric viewpoint},
    SERIES = {Encyclopaedia of Mathematical Sciences},
    VOLUME = {87},
      NOTE = {Control Theory and Optimization, II},
 PUBLISHER = {Springer-Verlag, Berlin},
      YEAR = {2004},
     PAGES = {xiv+412},
      ISBN = {3-540-21019-9},
   MRCLASS = {93-02 (49-02 49K15 93B27 93C15)},
  MRNUMBER = {2062547},
MRREVIEWER = {Kevin\ A.\ Grasse},
       DOI = {10.1007/978-3-662-06404-7},
       URL = {https://doi-org.ezproxy.biblio.polito.it/10.1007/978-3-662-06404-7},
}

@book{Lubliner2008a,
  title={Plasticity theory},
  author={Lubliner, Jacob},
  year={2008},
  publisher={Courier Corporation}
}

@article{Cermelli2001a,
  title={Configurational stress, yield and flow in rate-independent plasticity},
  author={Cermelli, Paolo and Fried, Eliot and Sellers, Shaun},
  journal={Proceedings: Mathematics, Physical and Engineering Sciences},
  pages={1447--1467},
  year={2001},
  publisher={JSTOR}
}

@article{Giverso2012a,
  title={Modelling the compression and reorganization of cell aggregates},
  author={Giverso, Chiara and Preziosi, Luigi},
  journal={Mathematical medicine and biology: a journal of the IMA},
  volume={29},
  number={2},
  pages={181--204},
  year={2012},
  publisher={OUP}
}

@article{Ambrosi2019a,
  title={Growth and remodelling of living tissues: perspectives, challenges and opportunities},
  author={Ambrosi, Davide and Amar, Martine Ben and Cyron, Christian J and DeSimone, Antonio and Goriely, Alain and Humphrey, Jay D and Kuhl, Ellen},
  journal={Journal of the Royal Society Interface},
  volume={16},
  number={157},
  pages={20190233},
  year={2019}
}

@article{DiStefano2026b,
  title={Multi-scale mechanics and remodelling of focal adhesions and the extracellular matrix},
  author={Di Stefano, Salvatore and Florio, Giuseppe and Puglisi, Giuseppe and Fazio, Vincenzo and Penta, Raimondo and Ram{\'\i}rez-Torres, Ariel},
  journal={Proceedings of the Royal Society of London Series A: Mathematical, Physical and Engineering Sciences},
  volume={482},
  number={2332},
  year={2026},
  publisher={The Royal Society}
}

@article{Sadik2017a,
  title={On the origins of the idea of the multiplicative decomposition of the deformation gradient},
  author={Sadik, S. and Yavari, A.},
  journal={Math. Mech. Solids},
  volume={22},
  number={4},
  pages={771--772},
  year={2017}
}

@article{Garikipati2006a,
  title={Biological remodelling: stationary energy, configurational change, internal variables and dissipation},
  author={Garikipati, K and Olberding, JE and Narayanan, H and Arruda, EM and Grosh, K and Calve, S},
  journal={Journal of the Mechanics and Physics of Solids},
  volume={54},
  number={7},
  pages={1493--1515},
  year={2006},
  publisher={Elsevier}
}

@book{Epstein2007a,
  title={Material inhomogeneities and their evolution: A geometric approach},
  author={Epstein, Marcelo and Elzanowski, Marek},
  year={2007},
  publisher={Springer Science \& Business Media}
}

@article{Maugin1998a,
  title={Geometrical material structure of elastoplasticity},
  author={Maugin, GA and Epstein, M},
  journal={International Journal of Plasticity},
  volume={14},
  number={1-3},
  pages={109--115},
  year={1998},
  publisher={Elsevier}
}

@article{Epstein2015a,
  title={Mathematical characterization and identification of remodeling, growth, aging and morphogenesis},
  author={Epstein, Marcelo},
  journal={Journal of the Mechanics and Physics of Solids},
  volume={84},
  pages={72--84},
  year={2015},
  publisher={Elsevier}
}

@article{Rodriguez1994a,
  title={Stress-dependent finite growth in soft elastic tissues},
  author={Rodriguez, Edward K and Hoger, Anne and McCulloch, Andrew D},
  journal={Journal of biomechanics},
  volume={27},
  number={4},
  pages={455--467},
  year={1994},
  publisher={Elsevier}
}

@article{Ambrosi2011a,
  title={Perspectives on biological growth and remodeling},
  author={Ambrosi, Davide and Ateshian, Gerard A and Arruda, Ellen M and Cowin, SC and Dumais, J and Goriely, A and Holzapfel, Gerhard A and Humphrey, Jay D and Kemkemer, R and Kuhl, Ellen and others},
  journal={Journal of the Mechanics and Physics of Solids},
  volume={59},
  number={4},
  pages={863--883},
  year={2011},
  publisher={Elsevier}
}

@article{Latorre2018a,
  title={Critical roles of time-scales in soft tissue growth and remodeling},
  author={Latorre, Marcos and Humphrey, Jay D},
  journal={APL bioengineering},
  volume={2},
  number={2},
  year={2018},
  publisher={AIP Publishing}
}

@article{Latorre2020a,
  title={Modeling biological growth and remodeling: contrasting methods, contrasting needs},
  author={Latorre, Marcos},
  journal={Current Opinion in Biomedical Engineering},
  volume={15},
  pages={26--31},
  year={2020},
  publisher={Elsevier}
}

@article{Ciancio2008a,
author = {Ciancio, V. and Dolfin, M. and Francaviglia, Mauro and Preston, Serge},
year = {2008},
month = {02},
pages = {},
title = {Uniform Materials and the Multiplicative Decomposition of the Deformation Gradient in Finite Elasto-Plasticity},
volume = {33},
journal = {Journal of Non-Equilibrium Thermodynamics},
doi = {10.1515/JNETDY.2008.009}
}

@incollection{Preston2010,
  title={Material uniformity and the concept of the stress space},
  author={Preston, Serge and El{\.z}anowski, Marek},
  booktitle={Continuous Media with Microstructure},
  pages={91--101},
  year={2010},
  publisher={Springer}
}

@article{Humphrey2014a,
  title={Mechanotransduction and extracellular matrix homeostasis},
  author={Humphrey, Jay D and Dufresne, Eric R and Schwartz, Martin A},
  journal={Nature reviews Molecular cell biology},
  volume={15},
  number={12},
  pages={802--812},
  year={2014},
  publisher={Nature Publishing Group UK London}
}

@article{Chen2019a,
  title={Electrical stimulation as a novel tool for regulating cell behavior in tissue engineering},
  author={Chen, Cen and Bai, Xue and Ding, Yahui and Lee, In-Seop},
  journal={Biomaterials research},
  volume={23},
  number={1},
  pages={25},
  year={2019},
  publisher={BioMed Central London}
}

@article{Qi2022a,
  title={Recent progress in active mechanical metamaterials and construction principles},
  author={Qi, Jixiang and Chen, Zihao and Jiang, Peng and Hu, Wenxia and Wang, Yonghuan and Zhao, Zeang and Cao, Xiaofei and Zhang, Shushan and Tao, Ran and Li, Ying and others},
  journal={Advanced Science},
  volume={9},
  number={1},
  pages={2102662},
  year={2022},
  publisher={Wiley Online Library}
}

@article{Liu2024a,
  title={Programmable mechanical metamaterials: basic concepts, types, construction strategies—a review},
  author={Liu, Chenyang and Zhang, Xi and Chang, Jiahui and Lyu, You and Zhao, Jianan and Qiu, Song},
  journal={Frontiers in Materials},
  volume={11},
  pages={1361408},
  year={2024},
  publisher={Frontiers Media SA}
}

@article{Ambrosi2025a,
  title={The shape of the mitral annulus: A hypothesis of mechanical morphogenesis},
  author={Ambrosi, Davide and Deorsola, Luca and Turzi, Stefano and Zoppello, Marta},
  journal={Mathematics and Mechanics of Solids},
  volume={30},
  number={2},
  pages={356--371},
  year={2025},
  publisher={SAGE Publications Sage UK: London, England}
}

\end{document}